\documentclass[11pt]{article}
\usepackage[dvipsnames,table]{xcolor}
\usepackage{amssymb,amsmath, amsthm, setspace, color, graphicx}
\usepackage{dsfont}
\usepackage{mathpazo, mathptmx, flexisym}
\usepackage{mathrsfs}            % use for calligraphy fonts
\usepackage{fancyhdr}
\usepackage[justification=centering]{caption}
\usepackage{subcaption}
\usepackage[authoryear]{natbib}
\usepackage{url}
\usepackage{breqn}
\usepackage{functan}
\usepackage[margin=1in]{geometry} % set margin 1 all around"
\usepackage[all, cmtip]{xy}
\usepackage[english]{babel}
\usepackage{graphicx}             % graphics package
\usepackage{float}
\usepackage[section]{placeins}
\usepackage{lscape} % for landscape
\usepackage{mathtools}
\usepackage{calc}     % for use in aligment of
\usepackage{enumitem} % descriptions
\usepackage[normalem]{ulem} % underline text
\usepackage[titletoc,toc,page]{appendix}% [titletoc,toc,page]
\usepackage{chngcntr}
\usepackage{etoolbox}
\usepackage{ragged2e}
\definecolor{mred}{rgb}{1.00,0.00,0.50}
\definecolor{mcyan}{rgb}{0.20,0.55,0.75}
\usepackage[pdfborder={0 0 0},final=true,colorlinks=true,breaklinks = true,linkcolor=mred,citecolor=mcyan]{hyperref}
\usepackage[capitalize,nameinlink,noabbrev]{cleveref}
\usepackage{relsize}
\usepackage{bigints}
\numberwithin{equation}{section}
\newcommand{\tab}{\hspace*{2em}}
\theoremstyle{plain}
\newtheorem{thm}{Theorem}[section]
\newtheorem{cor}{Corollary}
\newtheorem{prop}{Proposition}
\newtheorem{lem}{Lemma}
\newtheorem{assumption}{Assumption}      % The next few lines control assumption formats

\theoremstyle{definition}
\newtheorem{defn}{Definition}
\theoremstyle{remark}
\newtheorem{rem}{Remark}[section]

\providecommand{\red}[1]{\textcolor{red}{#1}}

\crefname{thm}{theorem}{theorems}
\Crefname{thm}{Theorem}{Theorems}
\crefname{cor}{corollary}{corollaries}
\Crefname{cor}{Corollary}{Corollaries}
\crefname{prop}{proposition}{propositions}
\Crefname{prop}{Proposition}{Propositions}
\crefname{lem}{lemma}{lemmas}
\Crefname{lem}{Lemma}{Lemmas}
\crefname{assumption}{assumption}{assumptions}
\Crefname{assumption}{Assumption}{Assumptions}
\crefname{defn}{definition}{definitions}
\Crefname{defn}{Definition}{Definitions}
\usepackage{titling}
\newcommand{\thesubtitle}{}

\posttitle{%
	\par\large\bfseries\thesubtitle
\end{center}
\vskip0.5em
}
\DeclareMathSizes{12}{12}{8}{6}
\DeclareMathOperator{\half}{\frac{1}{2}}

\begin{document}
\singlespace
\title{\textbf{Bank Run Exposure in a Paycheck-to-Paycheck Economy with Loss-Averse Depositors}}
\author{G. Charles-Cadogan\thanks{School of Accounting and Finance, University of Leicester; Tel: +44 (0116) 229 7385; e-mail: \textcolor[rgb]{0.00,0.00,1.00}{\href{mailto:gcc13@le.ac.uk}{gcc13@le.ac.uk}}~\\~\\
		The author has no competing interests or conflict of interest involving this paper.%\\~\\		
		%AI Disclosure Statement.
		%The author used generative AI tools to help spell check the document, to check the accuracy and consistency of mathematical statements and proofs, to help with debugging code. No AI system was used for hypothesis formation, model development or experimental design. The author independently verified all outputs and assume full responsibility for the accuracy and integrity of the results.
	}\\Working Paper}
\date{\vspace{-2ex}\today\vspace{-4ex}}
%\author{ Godfrey Cadogan \thanks{ Corresponding address: 3401-B N.W. 72nd Ave, Ste \# T-419, Miami, FL 33122; e-mail: \textcolor[rgb]{0.00,0.00,1.00}{\href{mailto:gocadog@gmail.com}{gocadog@gmail.com}}. I thank Mark Zanecki for comments, and Jan Kmenta for the idea of embedding parameters of interest in economic variables. I am indebted to Art Trager, at the Cowles Foundation, Yale University for yeoman library services without which this work would not be possible. Any error which may remain are my own.} \\ Working Paper}
%\date{\today}
%\date{}
%\long\def\symbolfootnote[#1]#2{\begingroup%
	%\def\thefootnote{\fnsymbol{footnote}}\footnote[#1]{#2}\endgroup}
%and use \symbolfootnote[1]{footnote} to get an *
\renewcommand\thefootnote{\fnsymbol{footnote}}
\maketitle
\renewcommand\thefootnote{\arabic{footnote}}
%\addtocounter{footnote}{-1}
%\large
\begin{abstract}
	\noindent We develop a behavioral model of bank run exposure in a paycheck-to-paycheck economy with loss-averse depositors. Income is received through demand deposits, and consumption ratcheting embeds reference dependence in a parsimonious asset-pricing framework. We show that sufficiently high subjective bad-state probabilities endogenously increase liquidity demand and generate equilibrium stress states supporting bank runs. These states define a Bank Run Exposure State Space and yield a martingale representation for exposure dynamics. A proof-of-concept empirical implementation using Call Report data constructs bank-level exposure proxies from funding and lending composition. A regression-weighted composite measure modestly improves fit relative to a retail-share benchmark, with stronger amplification among small banks and during the post-Silicon Valley Bank (SVB) collapse period. The framework highlights how behavioral liquidity demand alters equilibrium reserve holdings and can crowd out productive lending.
\end{abstract}
\singlespace
%\long\def\symbolfootnote[#1]#2{\begingroup%
	%\def\thefootnote{\fnsymbol{footnote}}\footnote[#1]{#2}\endgroup}
%and use \symbolfootnote[1]{footnote} to get an *
\renewcommand\thefootnote{\fnsymbol{footnote}}
\renewcommand\thefootnote{\arabic{footnote}}
%\addtocounter{footnote}{-1}
\large
\newpage
\thispagestyle{empty}
\tableofcontents
\listoffigures
\listoftables
\newpage
\thispagestyle{empty}
%\listoftables
\newpage
\pagenumbering{arabic}
\renewcommand\thefootnote{\fnsymbol{footnote}}
\renewcommand\thefootnote{\arabic{footnote}}
%\doublespace
\section{Introduction}\label{sec:Intro}
\tab Basel III regulates liquidity as if deposit outflows are exogenous stress parameters, yet in a paycheck-to-paycheck, digitally connected economy, liquidity demand is increasingly behavioral, state-dependent, and endogenously amplified. It addresses liquidity risk through the Liquidity Coverage Ratio (LCR), which requires covered banks to hold sufficient high-quality liquid assets to withstand net cash outflows over a 30-day stress scenario.\footnote{\href{https://www.bis.org/publ/bcbs238.pdf}{Basel Committee on Banking Supervision, ``Basel III: The Liquidity Coverage Ratio and liquidity risk monitoring tools,'' Bank for International Settlements, January 2013.}} In the United States, the LCR is implemented for Board-regulated institutions through Regulation WW, codified at \href{https://www.ecfr.gov/current/title-12/chapter-II/subchapter-A/part-249}{12 CFR Part 249}.\footnote{\href{https://www.ecfr.gov/current/title-12/chapter-II/subchapter-A/part-249}{12 CFR Part 249, ``Liquidity Risk Measurement, Standards, and Monitoring (Regulation WW),'' Board of Governors of the Federal Reserve System.}} Regulation WW sets out the liquidity coverage ratio, eligible high-quality liquid assets, total net cash outflow, and outflow amounts for different liability categories. For example, the retail funding outflow amount in \href{https://www.ecfr.gov/current/title-12/section-249.32}{12 CFR 249.32} includes 3\% of stable retail deposits, 10\% of other retail deposits, 20\% of certain fully insured third-party retail deposit placements, and 40\% of comparable third-party placements that are not fully insured; the same section assigns 100\% outflow treatment to certain unsecured wholesale funding that is not otherwise described.\footnote{See \href{https://www.ecfr.gov/current/title-12/section-249.32}{12 CFR 249.32, ``Outflow amounts.''}} These concrete regulatory factors make clear that the LCR already differentiates across funding stability, depositor type, insurance coverage, and wholesale funding risk.

\tab Those run-off factors are indispensable for supervision, but they are not a behavioral model of depositor liquidity demand. They summarize regulatory assumptions about average outflows under stress; they do not endogenize how depositor fear, paycheck timing, loss aversion, or digital access can amplify withdrawals in a particular state. This distinction matters because digitally enabled deposits are now part of the practical bank-run environment. Online banking, mobile banking, instant account access, and direct deposit make it easier for paycheck stress to become withdrawal pressure quickly. Accordingly, the model and empirical specification below treat digital access as a salient characteristic of bank clientele: digital channels do not create loss aversion, but they can accelerate the conversion of loss-averse liquidity demand into deposit outflows.

\tab The empirical motivation is visible in PayrollOrg's 2025 \href{https://info.payroll.org/pdfs/npw/2025-Getting-Paid-In-America-Survey-Results-Report.pdf}{``Getting Paid in America'' Survey Results Report}.\footnote{\href{https://info.payroll.org/pdfs/npw/2025-Getting-Paid-In-America-Survey-Results-Report.pdf}{PayrollOrg, ``2025 Getting Paid in America Survey Results Report.''}} The survey reports that 92.65\% of respondents receive pay by direct deposit.\footnote{See also \href{https://www.nacha.org/index.php/news/survey-direct-deposit-leads-way-payday}{Nacha, ``Survey: Direct Deposit Leads the Way on Payday,'' \textit{Nacha.org}, September 22, 2025.}} It also reports that 77.67\% would find it somewhat or very difficult to meet current financial obligations if the next paycheck were delayed by one week.\footnote{See also \href{https://www.cpapracticeadvisor.com/2025/09/18/if-your-paycheck-was-delayed-a-week-would-you-be-in-financial-difficulty/}{``If Your Paycheck Was Delayed a Week, Would You Be in Financial Difficulty?'' \textit{CPA Practice Advisor}, September 18, 2025.}} If a paycheck were delayed, respondents' primary ways to cover expenses would include delaying bills or payments (29.01\%), using savings (26.26\%), or borrowing from family and friends (25.46\%).\footnote{For secondary coverage of these paycheck-delay responses, see \href{https://metroatlantaceo.com/news/2025/09/survey-reveals-ongoing-financial-strain-majority-americans-depend-timely-paychecks/}{``Survey Reveals Ongoing Financial Strain as Majority of Americans Depend on Timely Paychecks,'' \textit{Metro Atlanta CEO}, September 19, 2025.}} The same survey reports that 73.64\% prefer higher wages to better health benefits, 54.96\% prefer extra cash in each paycheck to a large tax refund, and 77.70\% of respondents receiving annual raises said those raises were not enough to keep up with inflation. These results are consistent with broader reporting on paycheck-to-paycheck balance sheets and low household savings.\footnote{See \href{https://www.acainternational.org/news/income-volatility-reshapes-the-how-consumers-save/}{PYMNTS Intelligence, ``Income Volatility Reshapes How Consumers Save,'' \textit{ACA International}, November 21, 2025}; and \href{https://asset-sync.advisorperspectives.com/commentaries/2026/06/03/records-tape-savings-three-year-low}{``Records on the Tape: Savings at a Three-Year Low,'' \textit{Advisor Perspectives}, June 3, 2026.}}

\tab Taken together, these facts describe a paycheck-sensitive economy in which income timing, direct deposit, digital access, and inflation pressure can translate quickly into liquidity demand. Because paychecks are overwhelmingly received through demand deposit accounts, paycheck sensitivity creates a direct channel from income uncertainty to bank liquidity demand. This paper argues that LCR-type run-off factors may underestimate liquidity demand in behavioral stress states if they omit loss aversion and clientele-specific liquidity preference. When loss-averse depositors assign sufficiently high probability to a bad income state---states in which a delayed or reduced paycheck threatens immediate consumption needs---withdrawal demand can spike beyond average regulatory assumptions. This creates a \emph{behavioral liquidity externality}: banks must hold additional cash reserves against amplified withdrawal risk, reducing funds available for positive-net-present-value lending. Our model formalizes this channel and characterizes the stress states---the \textbf{Bank Run Exposure State Space}---where loss aversion transforms ordinary liquidity demand into a systemic threat to convertibility.

\tab The proof-of-concept empirical exercise implemented in the paper takes this argument to public FFIEC Call Report data. Because public quarterly filings cannot observe depositor-level paycheck timing, overdrafts, failed payments, or direct-deposit disruptions, the paper treats the exercise as an operational demonstration rather than a definitive test of the behavioral model. We first estimate a bank run exposure regression using retail share as a coarse proxy for paycheck-sensitive clientele. We then construct a regression-weighted composite proxy, $LA^{reg}_{bt}$, from transaction-deposit share, core-deposit share, and consumer-loan share, and compare it with both retail share and an equal-weight version of the composite. We also construct a small-bank indicator for institutions with less than \$1 billion in total assets and only domestic offices, and interact that indicator with liquidity risk and the exposure proxies. The comparison asks whether balance-sheet proxies designed to summarize several channels of depositor sensitivity carry more empirical content than retail share alone, including in event-window specifications around the Silicon Valley Bank episode and in small-bank subsamples where depositor relationships may be more local.

\subsection{Positioning Within the Literature}\label{subsec:LitReview}
\tab The literature on bank runs is huge. Of necessity, we sample pertinent parts of that literature that may be contrasted with our model. Arguably, the most influential paper on bank runs is \citet{DiamondDybvig1983}, hereinafter ``DD", who used a contract theory approach in a Bayesian-Nash model to distinguish among multiple equilibria--one of which supports a bank run. Specifically, Diamond and Dybvig specified two types of consumers with different time preference for consumption in a three (3) periods model. Under full information, first best optimal risk sharing entails each consumer satisfying her demand for consumption based on her time preference. By contrast, under asymmetric information agents withdraw from the bank for non-consumption reasons. Because banks lend long and borrow short, this leads to a second best unstable equilibrium with lurking bank runs.

\tab A short and comparatively contemporaneous paper by \citet{PostlewaiteVives1987} extended the DD model to four (4) periods in which they obtain a unique equilibrium based on a Prisoner`s dilemma paradigm. There, they imposed restrictions on returns to demand deposit and posited non-consumption based withdrawals as the source of potential runs. A key innovation in their paper is derivation of an endogenous probability of bank runs based on comparative statics of equilibrium.

\tab \citet[pg.~1294]{GoldsteinPauzner2005} critique of the DD model states ``that it does not provide tools to predict which equilibrium occurs or how likely each equilibrium is". Consequently, they modified the DD model by assuming that ``economic fundamentals are stochastic". That is, each agent has a production technology with state dependent probability of \emph{ex ante} output, and receives a private signal concerning future states of the economy. Range based signalling is enough to trigger a bank run. For example, the bank solves an expected utility maximization problem over the set of all possible states, based on a fixed payment for early withdrawal of demand deposits, \emph{ibid} pg 1308, and its surmise of \emph{taboo} limits for private signals of its customers. The first order condition for the solution of the bank's problem contains an implicit endogenous probability of bank runs.

\tab \citet{Uhlig2009} introduced a bank run model that explained aspects of the Great Recession of 2008. He acknowledged that the DD model was a benchmark, and argued that it was inapplicable to the 2007--2009 financial crisis because the composition of bank portfolios was different than DD assumed. For instance, Uhlig noted the high degree of mortgage backed securities and credit default swaps held in bank portfolios were a major cause of the crisis. In a nutshell, Uhlig's model suggests that the Great Recession of 2008 was a run of banks on banks wherein ``local banks hold deposit contracts on core banks, who in turn use the market to obtain liquidity," Uhlig at page 5. Although he did not model loss aversion explicitly, he implied that banks were risk seeking over ``period-2 opportunity cost for providing period-1 resources per selling long term services from the perspective of individual core banks"--an empirical regularity consistent with loss aversion. Compare \citet[pp.~13-14]{Uhlig2009}.

\tab In related work, \citet{MartinSkeieThadden2014} extended a ``bank run" model to ``financial institutions" by analogy with a repo market. There, a tri-party agent or bank serves the dual role of market maker and clearing house. Borrowers, i.e. security dealers, hold [long] marketable securities which they use as collateral for short term loans from creditors, i.e. pension funds, in a veritable revolving credit scheme.\footnote{\citet{BernardoWelch2004} also examine runs on financial institutions in the context of market microstructure. However, they did not consider tri-party repo. In their model, fear of future liquidity shocks was the driving force behind runs on the market. However, liquidity shocks was not modeled explicitly.} For by definition, borrowers enter an agreement to repurchase the security-collateral as a specified price (usually higher) at a later date. The tri-party agent acts as an intermediary which coordinates the process by matching borrowers with creditors/lenders on the basis of borrowers portfolio needs, and gives them the flexibility of substituting different classes of securities as collateral. A run on the security dealer occurs if creditors lose confidence in the market, i.e. their ability to sell the security-collateral at a price that covers costs, and call in their loans without refinancing.  \citet[pg.~2]{MartinSkeieThadden2014} characterized their paper thus: `` [o]ur main goal is to exhibit and model [] similarities [with the DD model], and to highlight the fundamental differences between securities dealers that borrow in the repo market against marketable securities as collateral and commercial banks that borrow unsecured deposits and hold nonmarketable loan portfolios" (emphasis added).

\tab \citet{TrautmanVlahu2013} experimentally investigate strategic loan defaults as a coordination problem, examining how expectations about bank fundamentals and borrower repayment capacity influence repayment decisions. Analyzing borrower-level characteristics, they find that risk attitudes--especially loss aversion over financial losses--strongly and robustly shape repayment behavior. Loss-averse borrowers value the cash they currently hold more than the larger but uncertain future payoff contingent on the bank's survival. As a result, they are more likely to withhold repayment to avoid the immediate loss they would incur if the bank were to fail.

\tab \citet{Kiss_et_al_2022} provide important experimental evidence on how loss aversion influences depositor behaviour during bank runs, particularly regarding line formation and withdrawal decisions. Their experiment shows that loss aversion's effect reverses depending on the information environment--loss-averse depositors bid less when they cannot observe others (viewing bids as losses), but bid more to withdraw first when observability allows them to signal and coordinate (viewing failure to coordinate as a loss). This suggests that loss aversion may operate through different mechanisms depending on context: the loss of money spent to move first versus the loss of the coordination payoff. \citeauthor{Kiss_et_al_2022} argue that loss aversion should be incorporated into models of bank runs. That is one of the objectives of this paper.

\tab \citet{ChoiGoldsmithPinkhamYorulmazer2023} study contagion around the Silicon Valley Bank (SVB) run and show that uninsured deposits, unrealized losses in held-to-maturity securities, bank size, and cash holdings help explain cross-bank stress after SVB. \citet{FeinsteinHalajSojmark2024} develop a balance-sheet model of depositor-run risk in which held-to-maturity accounting can mask run-relevant vulnerabilities. These recent papers are important for our framing because they emphasize balance-sheet and funding vulnerabilities exposed by the 2023 regional-bank episode. Our paper is complementary: it holds those institutional vulnerabilities in view but asks how depositor loss aversion and paycheck liquidity stress can make a given clientele more run-prone.

\tab \citet{CorreiaLuckVerner2026b} provide important new evidence on the distinction between bank runs and bank failures. Using large language models applied to historical U.S. newspapers, they identify more than 3,000 bank runs from 1863 to 1934 and show that many runs did not culminate in bank failure. Their evidence suggests that weak fundamentals are central to whether a run becomes a failure, while strong banks can survive runs through liquidity support, public reassurance, interbank cooperation, examination, or suspension of convertibility. This evidence is directly relevant to our framework because it clarifies that a bank run is not equivalent to failure. Our model instead focuses on the ex-ante exposure state: the behavioral and liquidity conditions under which a depositor clientele becomes run-prone before suspension or failure is observed. The empirical companion paper, \citet{CharlesCadoganExAnteWIP2026}, develops the LLM-based stress-test and regression implementation of this exposure idea; the present paper provides the underlying behavioral theory.

\tab Those bank run models which follow the Diamond-Dybvig contract theory paradigm \citep{DiamondDybvig1983} use consumption based utility models as a conduit for agents' preferences and behavior. To the best of our knowledge, this paper is the first to introduce a bank run model based on (1) a martingale exposure representation\footnote{\citet{Carmona2004} introduced a large-economy ``equilibrium bank runs" model in the context of a Diamond-Dybvig framework. However, that is distinguished from our probability-theory approach in a paycheck-to-paycheck economy.}, and (2) a behavioral paradigm with consumption ratcheting\footnote{\citet{Dybvig1995} used the consumption ratcheting principle in a model with habit formation.} in a paycheck-to-paycheck economy\footnote{This is distinguished from payday lending because ``payday loans are not typically offered by depository institutions" \citep[pg.~169]{Stegman2007}.} with (3) loss aversion to reduction in consumption embedded in \emph{bank run exposure}. In particular, we extend Tobin's liquidity-preference analysis \citep{Tobin1958} to include loss aversion to fluctuations in liquidity, and show how bank run exposure is determined by liquidity needs and liquidity risk.

\tab Our paper produces several new results, including an econometric implication for biased estimation of bank run exposure when liquidity risk is measured without loss aversion. Using a paradigm different from \citet{GoldsteinPauzner2005}, we identify a Bank Run Exposure State Space: the subset of states in which loss aversion, bad-state probability weights, and paycheck liquidity stress jointly support a run.

\tab We also show how probabilities of bank runs can be computed endogenously by identifying a stopped bank run process.\footnote{This approach is distinguished from \citet[pg.~24]{Dermine2009}, who computed probabilities in the context of relaxed Basel II capital requirements based on the large-portfolio expansion of the \citet{Vasicek2002} loan portfolio model.} In \citet[pg.~1309]{GoldsteinPauzner2005}, the computation of endogenous probability of bank runs is implicit and based on a first order condition for the bank's expected utility maximization problem. Here, the stopped-process approach links the run probability to the withdrawal sequence, the reserve constraint, and the subjective bad-state probability.

\tab Using a paradigm different from \citet{BowmanMinehartRabin1999}, we show how consumption based loss aversion is propagated in our model, and use that mechanism to characterize the exposure state space. Moreover, we develop a dynamic bank run exposure process and show that bank run exposure has a martingale representation.

\tab We then translate the exposure equation into a practical regression specification and illustrate the diagnostic content of the regression with simulated bank-run exposure plots. The empirical component of this research agenda is developed in the companion work in progress, \citet{CharlesCadoganExAnteWIP2026}, which uses an LLM-based stress-test design to implement and evaluate the proposed bank-run exposure regression when internal bank account-level data are difficult to obtain.

\tab Furthermore, we model stochastic loss aversion as a positive half-Cauchy process \citep{CharlesCadogan2018}, obtained by mapping a stationary Gaussian Ornstein-Uhlenbeck diffusion \citep[Ch.~15]{KarlinTaylor1981,Oksendal2003} through the probability-integral-transform and inverse half-Cauchy quantile construction \citep{Rosenblatt1952,Billingsley1995,Devroye1986,JohnsonKotzBalakrishnan1994}. This means that loss aversion evolves over time as a smooth Gaussian diffusion that is transformed via a quantile mapping so that its marginal distribution is positive and heavy-tailed (half-Cauchy), allowing for occasional extreme spikes in fear. This keeps the median loss-aversion index near $2.25$ while allowing extreme tail realizations that are central to run exposure.

\tab The rest of this paper proceeds as follows. In \cref{sec:TheModel} we introduce the model. There, in  \cref{subsec:LiquidityPreferenceTheory} we refine Tobin's liquidity-preference theory \citep{Tobin1958} to include loss aversion from fluctuations in income. The main result is summarized in \Cref{prop:LiquidityPrefEquation} and \Cref{cor:RiskSeekingOverIncomeLoss}. \cref{subsubsec:ConsumptionRatchet} introduces consumption ratcheting, and loss aversion to reduction in income and consumption. The main result there is \Cref{thm:LossAversionBasedMPC}. In \cref{subsec:OptimalConsumptionProblem} we introduce a simple consumption based behavioral asset pricing model. The main result there is \Cref{lem:BankRunTriggerEndogenLossAver}, which describes a sufficient stress-state condition under which a bank run can be triggered by loss aversion. \cref{sec:BankRunTopology} introduces the Bank Run Exposure State Space generated by those stress-state diagnostics. The main results are summarized in \Cref{thm:BankRunTopology}. In \cref{sec:ProbEstForStoppedBankRun} we provide probability estimates for suspension of convertibility based on a stopped bank run process. The main results there are \Cref{prop:BankRunProbCompute} and \Cref{cor:BankRunPoissonDist}. \cref{sec:BankRunExposeProcess} provides a martingale representation of bank run exposure, a practical exposure-regression specification, and simulated regression diagnostics. \cref{sec:EmpiricalImplementation} tests aspects of the theory with coarse bank call reports data. Finally, \cref{sec:Conclusion} concludes.
\section{The Model}\label{sec:TheModel}
\tab \citet[pp.~73-74]{Tobin1958} hypothesized a risk return tradeoff in which \emph{risk lovers} have negatively sloped indifference curves, while \emph{risk averters} have positively sloped indifference curves in (risk, return) space. Inasmuch as the level of Tobin's indifference curve represents a given quantum of utility attained by different combinations of risk and return, implicit in his hypothesis is the existence of otherwise convex utility associated with risk lovers, and concave utility associated with risk averters \citep[cf.][]{Markowitz1952}. By the same token, prospect theory posits that subjects are risk seeking over losses and risk averting over gains.\footnote{\citet[pg.~268]{KahnYver1979} refer to this phenomenon as the \emph{reflection effect}.} Moreover, risk seeking and risk averting behavior is observed in the same subject. \citet{KahnYver1979} posit a piecewise value function that is convex over losses and concave over gains wherein gains and losses are measured relative to a reference point in lieu of a utility function over levels \citep{VonNeumanMorgenstern1953}. Accordingly, we extend Tobin's analysis \citep{Tobin1958} to include loss aversion by and through the value-function parametrization in \citet{KahnYver1979}. In our model, loss aversion is endogenized as a result of the agent's assignment of probability weights to uncertainty about concomitant shocks to consumption and fluctuations to income.
\subsection{Liquidity preference as behavior towards loss aversion}\label{subsec:LiquidityPreferenceTheory}
\tab Let $x$ be the change in wealth levels--relative to a reference wealth level--for depositing clientele, and $v(x)$ be the corresponding value function for a representative depositor--in concert with prospect theory \citep{KahnYver1979,TverKahn1992}. \citet[pg.~9]{BarberisHuangSantos2001} note that $x$ could represent fluctuation(s) in income or wealth relative to a benchmark.\footnote{See also \citep[pg.~16]{Friedman1957}, who makes the case for relative wealth as opposed to levels of wealth in his ``permanent income" hypothesis.} For instance, if $Y(t)$ is disposable income at time $t$, and $CPI_{benchmark}(t)=\frac{CPI(t)}{CPI_{baseyear}}$, then $x(t)=\frac{CPI_{benchmark} (t)}{Y(t)}$ . Thus, $v(x)$ is the utility of wealth fluctuations.\footnote{\citet[pg.~6]{BarberisHuangSantos2001} used a more elaborate specification: $v(X_{t+1},S_t,z_t)$ where $X_{t+1}$ is fluctuation in investment, $S_t$ is risky assets held at time $t$, and $z_t$ is a state variable which measures past gains or losses as a fraction of $S_t$. However, in our model there is no intertemporal tradeoffs even though agents care about past gains or losses in income.}
\begin{assumption}
	Agents have liquidity preference paycheck to paycheck.
\end{assumption}
\begin{assumption}
	Agents are loss averse to fluctuations in wealth and income.
\end{assumption}
\tab These assumptions are consistent with Tobin's liquidity-preference theory \citep{Tobin1958}, i.e. agents live from paycheck-to-paycheck. Let $v_g(x)$ be the value function for gains, i.e. $x>0$. For losses, write $\ell=-x>0$ and let $v_\ell(\ell)$ denote the positive magnitude of the loss branch. Thus $v_g^\prime>0,\;v_g^{\prime\prime}<0,\;v_\ell^\prime>0$, and $v_\ell^{\prime\prime}<0$; the realized prospect-theory value of a loss is negative because the loss branch enters with a minus sign and is multiplied by the loss-aversion index $\lambda$. Let $\mathbb{I}_{A}$ be an indicator function defined over some set $A$. Then the value function can be written as
\begin{equation}\label{eq:ValueFuncFlucWealth}
	v(x)=\mathbb{I}_{\{x\geq0\}}v_g(x)-\mathbb{I}_{\{x<0\}}\lambda v_\ell (-x)
\end{equation}
\noindent where $\lambda$ is a loss aversion index.  \textit{In the sequel $v_\ell(-x)$ implies that $x<0$ in loss domain unless states otherwise.}  Suppose that changes in wealth levels is normally distributed, i.e. $x\sim N(\mu_x,\sigma_x)$, where $\mu_x$ is the mean change in wealth, and $\sigma_x$ is the corresponding risk. Let $\Phi(\cdot)$ be the cumulative normal distribution function and define the standardized variable $\mathfrak{z}$ such that $x=\mu_x+\mathfrak{z}\sigma_x$.
\begin{defn}[Radon measure] \label{def:RadonInt}
	Following \citet[pg.~114]{HewittStrom1965}, let $\mathcal{B}(X)$ be the space of all bounded functions defined on a nonempty set $X$. Let $\mathcal{C}_0(X) \subset \mathcal{B}(X)$ be the subspace of continuous functions with compact support in $X$, i.e. if $K \subset X$ is compact, then for some measure $\mu$ we have $\mu(X \backslash K) < \epsilon$ for some $\epsilon > 0$. A Radon measure is a nonnegative linear functional $I$ defined on $\mathcal{C}_0(X)$ if for $f,\;g \in \mathcal{C}_0(X)$ we have:
	\begin{enumerate}
		\item[i.~~] $I(f+g) = I(f) + I(g)$;
		\item[ii.~] $I(\alpha f) = \alpha I(f)$, for some scalar $\alpha$;
		\item[iii.] $I(f) \geq 0, \; \text{if}\; f \in \mathcal{C}_0^+(X)$.
	\end{enumerate}\qed
\end{defn}
\tab Furthermore, let $w(d\Phi(\mathfrak{z}))$ be a Radon measure of the probability weight assigned to a given change in wealth level according to cumulative prospect theory.\footnote{\citet[pg.~79]{BenartziThaler1995} state that under cumulative prospect theory ``one transforms cumulative rather than individual probabilities. Consequently, the decision weight [] depends on the cumulative distribution of the gamble". Our Radon measure approach to probability is consistent with \citet{ThomasVolcic1989}, and distinguished from \citet{Schmeidler1989}. The technical vagaries of those paradigms are outside the scope of this paper.} In the sequel we define
\begin{equation}\label{eq:ProbWeightPhiz}
	w(d\Phi(\mathfrak{z})) =
	\begin{cases}
		w^+(d\Phi(\mathfrak{z})) \quad\text{if gains}\\
		w^-(d\Phi(\mathfrak{z})) \quad\text {if losses}
	\end{cases}
\end{equation}
\noindent So that the expected value function is given by
\begin{equation}\label{eq:ExpectedValueFuncFlucWealth}
	V(\mu_x,\sigma_x)
	=E[v(x)]
	=
	\int^\infty_{-\infty}
	v(\mu_x+\mathfrak{z}\sigma_x)\;w(d\Phi(\mathfrak{z})).
\end{equation}

\tab Following \citet[pg.~75]{Tobin1958}, we interpret $V(\mu_x,\sigma_x)$ as defining a local indifference contour in mean-risk space. Under the maintained assumption that the probability weighting measure $w(d\Phi(\mathfrak{z}))$ is fixed with respect to local changes in $(\mu_x,\sigma_x)$, and that differentiation under the integral is valid, the compensating change in mean wealth per unit change in wealth risk is
\begin{equation}\label{eq:TobinProspectSlope}
	\frac{d\mu_x}{d\sigma_x}
	=
	-
	\frac{
		\int^\infty_{-\infty}
		\mathfrak{z}v^\prime(x(\mathfrak{z}))w(d\Phi(\mathfrak{z}))
	}{
		\int^\infty_{-\infty}
		v^\prime(x(\mathfrak{z}))w(d\Phi(\mathfrak{z}))
	}.
\end{equation}

\tab The derivation of \eqref{eq:TobinProspectSlope}, and the sign calculation behind the lemma below, are collected in \cref{sec:Proofs}. The gain and loss regions in \eqref{eq:TobinProspectSlope} are determined by the sign of $x(\mathfrak{z})=\mu_x+\mathfrak{z}\sigma_x$, not by the sign of $\mathfrak{z}$ alone. Hence,
\begin{equation}\label{eq:ProspectMarginalValue}
	v^\prime(x(\mathfrak{z}))
	=
	\mathbb{I}_{\{x(\mathfrak{z})\geq 0\}}
	v_g^\prime(x(\mathfrak{z}))
	+
	\mathbb{I}_{\{x(\mathfrak{z})<0\}}
	\lambda v_\ell^\prime(-x(\mathfrak{z})).
\end{equation}
\noindent The derivative is understood almost everywhere. Because $x(\mathfrak{z})=0$ occurs only at $\mathfrak{z}=-\mu_x/\sigma_x$, the kink in the prospect-theory value function does not create a mass-point problem under the maintained continuous distribution for $\mathfrak{z}$.

\begin{lem}[Determinants of sign of income-income risk tradeoffs]\label{lem:SignOfIncomeIncomeRiskTradeoff}
	The sign of tradeoffs between income fluctuation and income fluctuation risk is determined by
	\begin{equation*}
		\psi(\lambda)=g(\cdot)\biggl(\frac{\partial I_1}{\partial x}+\lambda\frac{\partial I_2}{\partial x}\biggr)-(I_1+\lambda I_2)(-J_1+\lambda J_2),
	\end{equation*}
	\noindent where $g,I_1,I_2,J_1$, and $J_2$ are defined in the proof of this lemma in \cref{sec:Proofs}.\qed
\end{lem}
\tab Inspection of $\psi(\lambda)$ shows that the local sign condition can change as loss aversion changes. This implies that the curve in $(\mu_x,\sigma_x)$ space is nonlinear because its slope changes with loss aversion. In particular, define a local critical value by
\begin{equation}\label{eq:LambdaCriticalValue}
	\lambda_c=\inf\{\lambda>0:\psi(\lambda)<0\}.
\end{equation}
\noindent Then $\beta^\prime(x)<0$ whenever $\lambda>\lambda_c$ and the denominator $g^2$ is nonzero. This result has econometric implications for the functional form of tradeoffs between fluctuations in income and dispersion of those fluctuations. In particular, we must also include a separate specification for the evolution of loss aversion to avoid simultaneity problems. See, for example, \citet[pg.~20, Eq.~(17)]{BarberisHuangSantos2001} for a specification of loss aversion, and \citet[pg.~7]{White2001}, \citet[pg.~312]{DavidsonMacKinnon1999}, and \citet[\S15.2]{Greene2012} for econometric theory behind estimation procedures. The proof details for the following proposition and corollary are also given in \cref{sec:Proofs}.
\begin{prop}\label{prop:LiquidityPrefEquation}
	Let $x$ be change in income, and $v$ be a value function over gains and losses with loss aversion index $\lambda$. Furthermore, let $\mu_x$ and $\sigma_x$ be the average change in income and $\sigma_x$ be the corresponding risk. Then the tradeoff between liquidity and liquidity risk is given by
	\begin{equation}
		\mu_x \approx \beta(x;\lambda)\sigma_x
	\end{equation}
	\noindent where $\beta(x;\lambda)$ is the local liquidity risk exposure.\qed
\end{prop}
\begin{cor}\label{cor:RiskSeekingOverIncomeLoss}
	Under the sign conditions above, there exists a local critical loss aversion value $\lambda_c$ such that $\beta^\prime(x;\lambda)<0$ whenever $\lambda >\lambda_c$ and $g(\mu_x,\sigma_x;\lambda)\neq 0$.	\qed
\end{cor}
\subsection{Myopic loss aversion paycheck to paycheck}\label{subsec:MyopicLossAversion}
\tab \citet[pp.~74-75]{BenartziThaler1995} used a qualitative model of \emph{myopic loss aversion} to explain the equity premium puzzle.\footnote{\citet{MehraPrescott1985} popularized this puzzle by showing that the historic gap between returns on equities and bonds was too large to be explained by a version of a general equilibrium asset pricing model \citep{Lucas1978}. Those authors were prescient when they concluded ``a model which is not an Arrow-Debreu economy" may better explain the large equity premium.} In particular, \citet[pg.~74]{BenartziThaler1995} opined ``two factors contribute to an investor being unwilling to bear the risks associated with holding equities, loss aversion and a short evaluation period. We refer to this combination as \emph{myopic loss aversion}." Furthermore, they conducted separate simulation experiments by computing prospective utility over increasing return horizons for stocks, and bonds. The intersection of prospective utilities determines the evaluation period for their portfolios. In our set-up, agents evaluate their income each pay period. Because they live from paycheck to paycheck, there are no savings or intertemporal tradeoffs. Thus, they exhibit \emph{myopic loss aversion} paycheck to paycheck.

\subsubsection{A time separable state independent paycheck model}\label{subsubsec:TimeSeparablePaycheckModel}
\tab Let $\{c_0,c_1\}$ be the plan for changes-in-consumption relative to a reference consumption level for a representative agent in our paycheck to paycheck economy, where $c_{0\omega}$ and $c_{1\omega}$ be changes in consumption at dates $0$ and $1$ respectively.\footnote{\citet[pg.~10]{Friedman1957} states:
	\begin{quote}
		``If a consumer unit knows that its receipts in any one year are unusually high and if it expects lower receipts subsequently, it will surely tend to adjust its consumption to its "normal" receipts rather than to its current receipts."
	\end{quote}
	\tab This is tantamount to the reference-point hypothesis in \citet{KahnYver1979} applied to ``permanent consumption". That is, $c_0$ and $c_1$ are transitory ``adjustments" to ``permanent consumption".} And $\omega$ is the discrete state of nature\footnote{The analysis that follows is motivated by \citet[pg.~77]{LeRoyWerner2000}.} in the set $\Omega$ of all possible states of nature. Let $v(c_{0\omega},c_{1\omega})$ be a time separable value function. Suppose that at date-$0$ shocks to income are offset by changes in consumption $x$, and that agents are uncertain that changes in income at date-$1$ will be offset by changes in consumption. Let $\rho$ be a constant discount factor, and $\pi_\omega$ be the probability weight agents assign\footnote{We could have modeled $\pi_\omega$ as being dependent on a noisy signal $\tilde{\theta}$ received by agents about future states of the economy \citep[see, e.g.,][]{GoldsteinPauzner2005}. That is, $\pi_\omega(\tilde{\theta})=\pi_\omega+\varepsilon$ is an agent's forecast of $\pi_\omega$. While that would have led to a richer model, it is not clear whether the main results would have been appreciably different. Therefore, we eschew that paradigm and assume that the probability weight is somehow assigned to states of nature without loss of generality.} to the state $\omega\in\Omega$ defined as follows:
\begin{equation}\label{eq:ProbWeightPiOmega}
	\pi_\omega =
	\begin{cases}
		\pi^+_\omega \quad\text{if gains}\\
		\pi^-_\omega \quad\text{if losses}
	\end{cases}
\end{equation}
\begin{assumption}
	Agents in our economy have zero wealth.
\end{assumption}
\begin{assumption}
	Agents consumption pattern is consistent with the permanent income hypothesis.
\end{assumption}
\begin{assumption}
	Banks create liquidity by issuing demand deposits.
\end{assumption}
\begin{assumption}
	Agents receive their paychecks by direct deposit in their bank accounts.
\end{assumption}
\begin{assumption}
	The bank meets withdrawal demand from cash reserves and immediately available liquid assets. Demand deposits are the runnable liability, and all other balance-sheet items are held fixed for the two-period argument.
\end{assumption}
\tab \citet[pp.~1347-1348]{Manski2004} reports that in response to the question
\begin{quote}
	What do you think is the percent chance (or what were the chances out of 100) that your total household income, before taxes, will be less than $Y$ over the next 12 months?
\end{quote}
\tab posed in the \emph{Survey of Economic Expectations}, designed by Jeff Dominitz and him, there was substantial variation in the interquartile range $(Q)$ for respondents with the same median $(M)$ under a lognormal distribution assumption. This implies heterogenous loss aversion to reduction in income and concomitant reduction in consumption. This empirical regularity motivates our
\begin{assumption}[Permanent consumption fluctuation hypothesis]\label{assum:PermanentConsumeFluctuationHypothesis}
	Fluctuations in consumption $(c)$ are related to fluctuations in income $(x)$ through the equation
	\begin{equation*}
		c_i=\tilde{k}(\rho,d;\lambda)x_i,\quad\text{where $i=0,1$}\label{eq:PermanentConsumptionHypothesis}\\
	\end{equation*}
	\noindent $\tilde{k}(\rho,d;\lambda)$ is the loss-aversion-adjusted marginal propensity to consume. It depends on the discount rate $\rho$, a vector $d$ of demographic characteristics, and the loss-aversion index $\lambda$ \citep[see][pg.~14]{Friedman1957}\citep[see also][]{BowmanMinehartRabin1999}. The value function $v$ is specified separately by the gain-loss representation and is not determined by $d$.
\end{assumption}
\begin{assumption}[Production technology]
	Agents produce only for their own consumption. Thus, in concert with \citet[pg.~406]{DiamondDybvig1983}, we assume that there is no trade, and the price of total consumption in period $1$ is $\rho$.
\end{assumption}
\begin{assumption}
	Agents ratchet consumption.
\end{assumption}

\subsubsection{A Hall--CRRA benchmark for loss-aversion calibration}\label{subsubsec:ConsumptionRatchet}
\tab In concert with \citet[pg.~14]{Friedman1957}, let $k(\rho,d)$ be the marginal propensity to consume (``MPC"), where $\rho$ is a discount rate, and $d$ is a vector of demographic factors. Let $Y_{i,\omega},\quad i=0,1\quad \omega\in\Omega$ be income level(s), and $Y_p$ be permanent income.
\begin{defn}[Consumption ratchet]\label{def:ConsumptionRatchet}
	Let $\Omega$ be the sample space for states of nature, $\mathcal{F}$ be the $\sigma$-field of Borel measurable subsets of $\Omega$, and $P$ be a probability measure on $\Omega$. Let $C=\{C(t,\omega);\; 0\leq t<\infty,\;\omega\in\Omega\}$ be a real valued right continuous consumption process defined on the probability space $(\Omega,\mathcal{F},P)$ such that $C(0,\omega)=0$. Let $c(t,\omega)=C(t^+,\omega)-C(t^-,\omega)$. If for $\forall t,\quad c(t,\omega)\geq 0$, then $c(t,\omega)$ is called a \emph{consumption ratchet}.	\qed
\end{defn}
\begin{rem}
	This definition implies that the consumption process is a \emph{subordinator}, and that it may possibly have a unit root \citep[cf.][pg.~405]{KaratzaShreve1991}; see also \citet{Hall1978}.\qed
	
\end{rem}
\tab Suppose that the marginal propensity to consume is deterministic, income is subject to state dependent transitory shocks $\eta_{i\omega}$, and permanent income $Y_p$ is the same in each of the two time periods. The algebra connecting transitory income shocks, the loss-aversion-adjusted marginal propensity to consume, and the one-period ratchet increment is collected in \cref{subsec:ProofsMyopicLossAversion}. The result is that ratcheted consumption is nonnegative precisely when the marginal propensity to consume with loss aversion is at least as large as the marginal propensity to consume without it. This gives rise to the following.
\begin{thm}[Loss Aversion Based Marginal Propensity To Consume]\label{thm:LossAversionBasedMPC}
	Let $Y_p>0$, let $\tilde{k}(\rho,d;\lambda(\omega))$ be the marginal propensity to consume with state dependent loss aversion, and let $k(\rho,d)$ be the marginal propensity to consume without it. Then the one-period ratchet increment $\{c_{1\omega}\}_{\omega\in\Omega}$ is nonnegative if and only if $\tilde{k}(\rho,d;\lambda(\omega))\geq k(\rho,d)$.\qed
	
\end{thm}
\begin{rem}
	This result is consistent with \citet{CarrollKimball1996}, who find that marginal propensity to consume is higher under uncertainty, albeit by a different approach.\qed
	
\end{rem}
\tab The purpose of the Hall comparison is calibration rather than a separate test of the permanent-income/random-walk hypothesis. The preceding ratchet result gives a state-dependent loss-aversion index, but the empirical and simulation exercises need a smooth benchmark around which stress-state spikes can occur. A Hall--CRRA consumption-growth factor provides that benchmark. If \eqref{eq:MarginalpropenConsumeWithLambdaYp} holds for separable loss aversion, then
\begin{align}
	C_{1\omega}&=\tilde{k}(\rho,d;\lambda(\omega))Y_p\\
	&=\tilde{k}(\rho,d)\lambda(\omega)Y_p\\
	\intertext{If $\tilde{k}(\rho,d)$ is Friedman`s MPC, then}
	C_{1\omega}&=\tilde{k}(\rho,d)\lambda(\omega)Y_p=\lambda(\omega)C_1\\
	\intertext{where $C_1$ is permanent consumption. When juxtaposed to \citet[pg.~975]{Hall1978} random walk hypothesis}
	C_{1\omega}&=C_{0\omega} + \epsilon_{1\omega}\label{eq:HallConsumpRandomWalk}
	\intertext{where, according to Hall, $\epsilon_{1\omega}$ ``summarizes the impact of all information that becomes available in period $t$ about the consumer`s lifetime well-being". It implies that}
	C_{1\omega}&=\lambda(\omega) (C_{0\omega}+\epsilon_{1\omega})\\
	&=\lambda(\omega)C_{0\omega} +\lambda(\omega)\epsilon_{1\omega}\\
	\intertext{However, Hall defined}
	\lambda^{\text{Hall}}&=\biggl(\frac{1+\delta}{1+r}\biggr)^{-\tfrac{1}{R(c)}}\\
	\intertext{where $\rho$ is a subjective discount rate; $r$ is real rate of interest $r\geq \delta$; and}
	R(c)&=-c\frac{u^{\prime\prime}(c)}{u^\prime(c)}\\
	\intertext{is Arrow-Pratt measure of relative risk aversion for a concave utility function $u(\cdot)$. Hall called it ``the elasticity of marginal utility". This comparison \textbf{suggests} an inverse relation between loss aversion and the Hall consumption-growth factor in low-growth states:}
	\lambda(\omega)&\propto\frac{1}{\lambda^{\text{Hall}}(\omega)}\label{eq:InverseLambdaHall}
\end{align}
\noindent The relationship in \eqref{eq:InverseLambdaHall} is a heuristic calibration motivated by the comparison between the Hall consumption-growth factor and the loss-aversion index. It provides a benchmark for calibrating $\lambda$ in terms of observables, rather than a rigorous proof of an inverse causal relationship. In the dynamic exposure section below, this benchmark becomes the median or locally fixed component of the stochastic loss-aversion process; the heavy-tailed component then represents rare fear states around that benchmark. We summarize this heuristic with the following
\begin{thm}[Consumption based loss aversion heuristic]\label{thm:ConsumptionBasedLossAver}
	Let $\lambda(\omega)$ be state dependent loss aversion in a paycheck to paycheck economy. Define
	\begin{align*}
		\lambda^{\text{Hall}}&=\biggl(\frac{1+\delta}{1+r}\biggr)^{-\tfrac{1}{R(c)}}\\
		\intertext{where $\rho$ is a subjective discount rate; $r$ is real rate of interest $r\geq \delta$; and}
		R(c)&=-c\frac{u^{\prime\prime}(c)}{u^\prime(c)}
	\end{align*}
	\noindent In particular, $\lambda^{\text{Hall}}$ is the growth rate of consumption in a Hall-type economy. Then, under the separability and proportionality conventions described above, the following heuristic calibration provides a benchmark for loss aversion:
	\begin{equation*}
		\lambda(\omega)\propto\frac{1}{\lambda^{\text{Hall}}(\omega)}.
	\end{equation*}
	\noindent That is, loss aversion is proportional to the inverse of the growth rate of consumption in this heuristic calibration.\qed
\end{thm}
\tab \textbf{\sc{Example}}. Suppose that $u(c)=c^{(\sigma-1)/\sigma}$ is CRRA utility as parametrized by \citet[pg.~974]{Hall1978}. Then $R(c)=1/\sigma$. Thus, substitution in the expression for $\lambda(\omega)$ implies that
\begin{align}
	\lambda(\sigma;r\;\delta)&=a_0\biggl(\frac{1+r}{1+\delta}\biggr)^{\sigma}\\
	\intertext{where $a_0$ is a constant of proportionality, and $\sigma$ is the elasticity of intertemporal substitution (EIS). After taking logs on both sides and using a first order approximation, we get}
	\lambda(\sigma)&=a_0\exp\biggl\{\sigma\biggl(\frac{r-\delta}{1+\delta}\biggr)\biggr\}
\end{align}
\noindent which requires $r\geq \delta$. Thus, we have a feedback equation between loss aversion and Arrow-Pratt relative risk measure for CRRA utility. For reasonable values of $a_0$, this equation comports well with estimates of $\lambda\approx e\approx 2.78$ for $\sigma,\;\delta$; and $ r$ chosen accordingly \citep[see, e.g.,][]{AbdellaBleichrodtParas2007}. In any event, the foregoing analysis shows how consumption based loss aversion is propagated in our model.

\tab The Hall/CRRA expression in \eqref{eq:HallConsumpRandomWalk}--\eqref{eq:MarginalpropenConsumeWithLambdaYp} provides the smooth structural component of loss aversion implied by consumption-growth fundamentals. To allow for latent depositor heterogeneity and rare fear states, we use the calibration
\begin{equation}\label{eq:LambdaHallHalfCauchyMultiplier}
	\lambda_t
	=
	\lambda_t^{\mathrm{Hall}}\xi_t,
	\qquad
	\lambda_t^{\mathrm{Hall}}
	=
	a_0\biggl(\frac{1+r_t}{1+\delta_t}\biggr)^{\sigma_t},
	\qquad
	\xi_t\sim |Cauchy(0,1)|.
\end{equation}
\noindent This multiplicative specification preserves positivity and keeps the Hall/CRRA calibration as the median component. Since the standard half-Cauchy multiplier has median one, $\operatorname{median}(\lambda_t\mid\lambda_t^{\mathrm{Hall}})=\lambda_t^{\mathrm{Hall}}$. Thus, the half-Cauchy component does not replace the consumption-based theory; it adds the heavy right tail needed to represent rare, economically important loss-aversion states. These are precisely the states that matter for bank-run exposure, because the model is concerned with stress realizations rather than only average depositor behavior. For the continuous-time representation, one may either hold $\lambda_t^{\mathrm{Hall}}$ fixed over a short monitoring window and write it as $m_\lambda$, or set $m_\lambda$ equal to the relevant bank-segment Hall/CRRA benchmark at the calibration date.

\vspace{0.5\baselineskip}
\begin{center}
	\red{\textbf{[Figure 1 goes about here: Calibration of the loss-aversion index]}}
\end{center}
\vspace{0.5\baselineskip}

\subsection{The optimal consumption problem}\label{subsec:OptimalConsumptionProblem}
\tab The agent solves the following problem.\footnote{The borrowing and discounting feature of our model distinguish it from \citet{BowmanMinehartRabin1999}, who assumed otherwise.}
\begin{align}
	V^*(c_0,c_1)&=\max_{c_0,c_1}E[v(c_{0\omega},c_{1\omega};\Omega)]\\
	\text{s.t.}\;c_{0\omega}&=x_{0\omega}\\
	c_{1\omega} &\geq x_{1\omega} \label{eq:CreditCardConstraint}
\end{align}
\begin{rem}
	Examination of \citet[pg.~191]{Diamond2007} exposition of the Diamond-Dybvig (``DD") model shows that the ``states of nature" in their model are two investor types--each with different liquidity horizons--and that consumption equals assets over each liquidity preference horizon. In our model, the ``type" space is the range space for loss aversion. Moreover, our agents want to liquidate their paycheck in the next pay period. Arguably, those aspects of the DD formulation are functionally equivalent to our consumption based asset pricing model but for value functions and probability weighting.\qed
	
\end{rem}
\tab The KKT derivation for this constrained problem is collected in \cref{subsec:ProofsOptimalConsumption}. It shows that a sufficient stress-state trigger occurs when the weighted marginal value of a loss state dominates the weighted marginal value of a gain state.
\begin{rem}
	\citet[pp.~579-580]{LoewensteinPrelec1992} introduced a discounted value function, with \emph{hyperbolic discounting} parametrized as
	\begin{align}
		V(x_1,t_1;x_2,t_2;\dotsc,x_n,t_n)&= \Sigma^n_{i=1}v(x_i)\phi(t_i)\\
		\intertext{where the discount function}
		\phi(t)&=\biggl(1+\alpha t\biggr)^{-\tfrac{\beta}{\alpha}}\\
		\intertext{where $\alpha$ is a shape parameter measuring departure from constant discounting, and}
		\lim_{\alpha\downarrow 0} \phi(t)&=e^{-\beta t}
	\end{align}
	\noindent In order not to overload the model we used a constant ``exponential" discount factor $\rho$.\qed
	
\end{rem}
\begin{rem}
	The constraint in \eqref{eq:CreditCardConstraint} implies that not only is our agent living from paycheck to paycheck but that positive shocks to consumption that exceed changes in her income induce loss aversion. Thus, she incurs credit card debt\footnote{See \citet[pp.~58-59]{AngeletosEtAl2001} for hyperbolic discounters' use of credit cards as a buffer to liquidity shocks, and \citet{GrossSouleles2002} for a taxonomy of empirical evidence on the use of credit cards for consumption smoothing, \emph{inter alia}.} or overdrafts that exacerbate the bank's cash reserves. In particular, our agent's uncertainty about shocks to consumption relative to changes in income is resolved in the next pay period.\qed
	
\end{rem}
\begin{rem}
	Implicit in our model is the existence of a state vector that includes aspects of the production sector excluded here \citep[see, e.g.,][pg.~269]{Breeden1979}.\qed
	
\end{rem}
\tab We summarize the results in this section with the following
\begin{lem}[Bank Run Triggered By Endogenous Loss Aversion]\label{lem:BankRunTriggerEndogenLossAver}
	Let $\tfrac{\pi^+_\omega}{\pi^-_\omega}$ be the representative agent's surmise of the odds of a positive to negative shock to next period's income $x_{1\omega}$. Let $\lambda(\omega)$ be state dependent loss aversion, and let $v_g$ and $v_\ell$ be the gain branch and positive loss-magnitude branch of the value function, respectively. A sufficient stress-state condition for a bank run in a paycheck to paycheck economy is
	\begin{equation}
		\lambda(\omega^\prime)>\frac{\pi^+_{\omega^\prime}v^\prime_g(c_{1\omega^\prime})}{\pi^-_{\omega^\prime}v^\prime_\ell(-c_{1\omega^\prime})}
	\end{equation}
	\noindent where $\omega^\prime\in\Omega$ is a state of uncertainty in which the agent solves her optimal consumption problem.\qed
	
\end{lem}
\section{Bank Run Exposure State Space driven by loss aversion}\label{sec:BankRunTopology}
\tab Let $\underline{\Omega}^{\text{admiss}}$ be the set of admissible stress states for which the inequality in \Cref{lem:BankRunTriggerEndogenLossAver} holds, and
\begin{equation}\label{eq:LambdaCritical}
	\lambda_c(\omega^\prime) = \frac{\pi^+_{\omega^\prime}v^\prime_g(c_{1\omega^\prime})}{\pi^-_{\omega^\prime}v^\prime_\ell(-c_{1\omega^\prime})}
\end{equation}
\noindent \citet[pg.~121]{KobberlingWakker2005} used the ratio of left and right derivatives at the reference point $c=0$ to define the loss aversion parameter with a L'Hospital type rule\footnote{See \citet[pp.~292-293]{Apostol1967} for details on this rule.} as
\begin{align}
	\lambda &= \lim_{c\downarrow 0} \frac{v_g^\prime(c)}{v_\ell^\prime(-c)}
	\intertext{In the context of \eqref{eq:LambdaCritical}  this implies that}
	\lambda(c_{1\omega^\prime}) &= \frac{v^\prime_g(c_{1\omega^\prime})}{v^\prime_\ell(-c_{1\omega^\prime})}
\end{align}
However, according to \Cref{thm:LossAversionBasedMPC} if loss aversion is separable in $\tilde{k}(\rho,d;\lambda(\omega))$, i.e. $\tilde{k}(\rho,d;\lambda(\omega)) = \tilde{k}(\rho,\;d)\lambda(\omega)$, and $\lambda \gg 1$ \citep[see, e.g.,][pg.~1662]{AbdellaBleichrodtParas2007}, and $\tilde{k}$ is a loss aversion adjusted extension--then
\begin{equation}
		\frac{k}{\tilde{k}}\leq \lambda(c_{1\omega^\prime}) < \frac{1}{\tilde{k}}
\end{equation}
So that
\begin{align}
	\lambda_c(\omega^\prime) &= \frac{\pi^+_{\omega^\prime}}{\pi^-_{\omega^\prime}}\lambda(c_{1\omega^\prime})\label{eq:LambdaLogOdds}\\
	\intertext{implies that}
	\lambda_c(\omega^\prime) &\in \frac{\pi^+_{\omega^\prime}}{\pi^-_{\omega^\prime}}\biggl[ \frac{k}{\tilde{k}},\;\frac{1}{\tilde{k}}\biggr)
	\intertext{and $\lambda > \lambda_c$ implies}
	\lambda &\notin \frac{\pi^+_{\omega^\prime}}{\pi^-_{\omega^\prime}}\biggl[ \frac{k}{\tilde{k}},\;\frac{1}{\tilde{k}}\biggr)\label{eq:TabooLimitsLossAver}\\
	\intertext{Let}
	\underline{\Omega}^{\text{taboo}} &= \biggl\{\omega^\prime\vert\; \lambda \notin \frac{\pi^+_{\omega^\prime}}{\pi^-_{\omega^\prime}}\biggl[ \frac{k}{\tilde{k}},\;\frac{1}{\tilde{k}}\biggr)\biggr\}\label{eq:OmegaTaboo}
\end{align}
\noindent So that we have the following
\begin{lem}[Admissible Bank Run States Induced By Loss Aversion]\label{lem:AdmissBankRunStates}
	Let $\pi^+_\omega,\;\pi^-_\omega$ be the agent's surmise of the probability for gains or losses, respectively, in state $\omega\in\Omega$, and $v_g,\;v_\ell$ be the value functions for gains and losses, and $\lambda$ be the loss aversion index. Then the set of admissible states for bank runs is given by
	\begin{equation}\label{eq:OmegaAdmiss}
		\underline{\Omega}^{\text{admiss}}=\biggl\{\omega^\prime\in\Omega:\;\lambda(\omega^\prime)> \frac{\pi^+_{\omega^\prime}v^\prime_g(c_{1\omega^\prime})}{\pi^-_{\omega^\prime}v^\prime_\ell(-c_{1\omega^\prime})} \biggr\}
	\end{equation}
	\qed
\end{lem}
\tab Undeniably, $\underline{\Omega}^{\text{admiss}}\subseteq\Omega\;\text{and}\;\underline{\Omega}^{\text{admiss}}\neq\emptyset$.
\tab From the foregoing we get the
\begin{lem}[Taboo states for loss aversion]\label{lem:TabooStatesLossAver}
	Let $v_g$ and $v_\ell$ be the gain branch and positive loss-magnitude branch, respectively, pasted at the reference point $x=0$ to get the value function in \eqref{eq:ValueFuncFlucWealth}. Let $\pi^+_\omega,\;\pi^-_\omega$ be the agent's surmise of the probability for gains or losses, respectively, in state $\omega\in\Omega$. Define
	\begin{align}
		\lambda &= \lim_{c\downarrow 0} \frac{v_g^\prime(c)}{v_\ell^\prime(-c)}\label{eq:KobberlingWakkerLambda}\\
		\intertext{Then a bank run is triggered in the following \emph{taboo states} for loss aversion}
		\underline{\Omega}^{\text{taboo}} &= \biggl\{\omega^\prime\vert\; \lambda \notin \frac{\pi^+_{\omega^\prime}}{\pi^-_{\omega^\prime}}\biggl[ \frac{k}{\tilde{k}},\;\frac{1}{\tilde{k}}\biggr)\biggr\}
	\end{align}
	\qed
\end{lem}
\tab Thus, the set of states that support a bank run is given by
\begin{equation}\label{eq:OmegaRun}
	\underline{\Omega}^{\text{run}} = \underline{\Omega}^{\text{admiss}}\cap\underline{\Omega}^{\text{taboo}}
\end{equation}
\noindent Inasmuch as $\lambda \gg 1$, \Cref{lem:TabooStatesLossAver} implies that for $\tfrac{\pi^+}{\pi^-}$ sufficiently small a bank run would be triggered. That is, an agent`s subjective probability of loss $(\pi^-)$ is much larger than her subjective probability of gain $(\pi^+)$. So that the set of taboo states $\underline{\Omega}^{\text{taboo}}$ is well defined. Thus, the loss aversion parameter $\lambda(\omega^\prime)$ is an endogenous random variable that depends on the subjective odds of gains to losses by and through $\lambda_c$ in \eqref{eq:LambdaCriticalValue}. We formalize this with the following
\begin{lem}[Endogenous loss aversion]\label{lem:EndogenousLossAversion}
	Loss aversion $\lambda$ is an endogenous random variable that depends on the subjective probability gains to losses in consumption $(c)$ arising from fluctuations in paycheck $(x)$.
\end{lem}
\tab These conditions on $\lambda$ imply that the direction of the slope of $\beta$ in \eqref{eq:BetaSlopeDirection} can turn downward for interior solutions corresponding to changes in consumption in states $\omega^\prime\in\underline{\Omega}^{\text{run}}$. Moreover, the support of $\underline{\Omega}^{\text{run}}$ rests on loss aversion induced by uncertainty over bad consumption shocks with comparatively large probabilities $\pi^-_{\omega^\prime}$ of occurrence. By contrast, in the DD-model bank runs are triggered if agents forecast of the fraction $f$ of depositors who withdraw at the beginning of period $1$ exceed a critical value determined by locally stable equilibrium \citep[see][pg.~197]{Diamond2007}.

\tab The preceding restrictions define the Bank Run Exposure State Space. Let $\Lambda\subset\mathbb{R}_+$ be a set of admissible loss-aversion indexes and let $\lambda:\underline{\Omega}^{\text{run}}\rightarrow\Lambda$ map a run state into its implied loss-aversion index. Equip $\Lambda$ with the topology generated by neighborhoods of the form
\begin{equation}\label{eq:MainStateSpaceNeighborhood}
	\mathcal{O}(\omega^\prime)=
	\biggl\{\lambda>0:\lambda>\lambda_c(\omega^\prime)\biggr\}
	\setminus
	\frac{\pi^+_{\omega^\prime}}{\pi^-_{\omega^\prime}}\biggl[ \frac{k}{\tilde{k}},\;\frac{1}{\tilde{k}}\biggr),
\end{equation}
\noindent where the set difference is interpreted in the usual sense. Then the induced structure on bank-run exposure states is the pullback topology
\begin{equation}\label{eq:MainPullbackStateSpace}
	\mathcal{T}^{\text{run}}=\{\lambda^{-1}(O):O\in\mathcal{T}_\Lambda\}.
\end{equation}
\noindent The formal basis, measurability, and random-field details are collected in the Internet Appendix. From the foregoing analyses we have the following
\begin{thm}[Bank Run Exposure State Space]\label{thm:BankRunTopology}
	The Bank Run Exposure State Space $(\underline{\Omega}^{\text{run}},\mathcal{T}^{\text{run}})$ is induced by loss aversion and uncertainty over comparatively large probabilities of bad shocks $\pi^-_{\omega^\prime}\gg \pi^+_{\omega^\prime}$ to fluctuations in income $(x)$ and consumption $(c)$.
\end{thm}
\noindent The proof is given in \cref{subsec:ProofsStateSpace}.
\begin{rem}
	\citet[pp.~1305-1307]{GoldsteinPauzner2005} also characterized a set of states that support bank runs based on bank clientele private signals about the state of the economy.\qed
	
\end{rem}
\tab Thus, the Bank Run Exposure State Space induced by clientele loss aversion can be indexed by time, depositor segment, geography, or employer group. When the index set is time, the associated loss-aversion collection becomes a stochastic loss-aversion process. This extension motivates the stopped-process probability estimates and exposure representation that follow.

\section{Probability estimates for suspension of convertibility}\label{sec:ProbEstForStoppedBankRun}
\tab The preceding Bank Run Exposure State Space identifies the stress states in which loss aversion can turn ordinary liquidity demand into bank-run exposure. This section translates that state-space description into a probability statement about suspension of convertibility. Following \citet{DiamondDybvig1983} and \citet{Diamond2007}, suspension of convertibility is interpreted as the point at which the bank must stop converting deposits into cash because sequential withdrawal demand exhausts available cash reserves.

\tab Let $\{\mathcal{G}_n\}_{n\geq 0}$ be the filtration generated by the sequential withdrawal order and associated income shocks. If the bank has $N$ customers with heterogeneous loss aversion, define the stopping time
\begin{equation}
	\tau_n(\omega;\lambda_1,\dotsc,\lambda_n) =
	\begin{cases}
		\min\{n>0:\; CR(\omega;\lambda_1,\dotsc,\lambda_n)<\Sigma^n_{i=1}x_{i\omega},\; \omega\in\underline{\Omega}^{\text{run}}\}&\text{$ n\leq N$}\\
		\infty&\text{otherwise.}
	\end{cases}
\end{equation}
\noindent Thus, suspension occurs when cumulative withdrawal demand first exceeds cash reserves in a run state. If stress withdrawals are Bernoulli events with subjective bad-state probability $\pi^-_{i\omega}$, then the number $K$ of non-stress transactions observed before the $\tau_n(\omega)$-th stress withdrawal has a negative-binomial representation. Writing
\begin{equation}
	\underline{\pi}^-_{\omega}=\inf_{1\leq i\leq N}\{\pi^-_{1\omega},\dotsc,\pi^-_{N\omega}\},
\end{equation}
\noindent gives the following probability statement.
\begin{prop}[Computation of bank run probabilities]\label{prop:BankRunProbCompute}
	If stress withdrawals are Bernoulli events and $\tau_n(\omega)$ is the stopping time at which cumulative withdrawal demand breaches cash reserves, then the probability of observing $k$ non-stress transactions before the stopped run threshold is
	\begin{equation}\label{eq:MainBankRunProb}
		\text{Pr}(K=k)=\binom{\tau_n(\omega)+k-1}{k}(\underline{\pi}^-_{\omega})^{\tau_n(\omega)}(1-\underline{\pi}^-_{\omega})^k.
	\end{equation}
\end{prop}
\tab For large stopping times and small $1-\underline{\pi}^-_{\omega}$, the corresponding Poisson approximation is
\begin{cor}[Poisson probability of bank run]\label{cor:BankRunPoissonDist}
	\begin{equation}\label{eq:MainPoissonRunProb}
		\text{Pr}(K=k)=\frac{\mu^k}{k!}e^{-\mu},\qquad
		\mu=\frac{\tau_n(\omega)(1-\underline{\pi}^-_{\omega})}{\underline{\pi}^-_{\omega}}.
	\end{equation}
\end{cor}
\tab The main implication is direct: a bank run is most likely when loss-averse depositors attach sufficiently high subjective probability to a bad income state and the stopped withdrawal order reaches the bank's reserve constraint. The full negative-binomial derivation and related stopped-process notation are collected in the Internet Appendix.

\section{Bank run exposure process}\label{sec:BankRunExposeProcess}
\tab The fundamental difference between the present model and \citet{Tobin1958} is that the risk-return tradeoff includes a loss-aversion parameter. Because that parameter is itself random, bank run exposure must be treated as a stochastic object rather than as a fixed slope coefficient. Let the aggregate paycheck-to-paycheck economy be summarized by
\begin{align*}
	\bar{x}_N(t,\omega) &=\frac{1}{N}\Sigma_{i=1}^N x_i(t,\omega),\\
	\sigma_{x,N}(t) &= \left(\frac{1}{N}\Sigma^N_{i=1}\sigma^2_{i,x}(t)\right)^\half.
\end{align*}
\noindent The corresponding linear exposure representation is
\begin{equation}\label{eq:MainLinearIncomeRiskModel}
	\bar{x}_N(t,\omega)
	=
	\beta_0 + \beta_{x,N}(\bar{\lambda}(t))\sigma_{x,N}(t)+\bar{\epsilon}_N(t,\omega).
\end{equation}
\noindent Market clearing implies that per-capita cash-reserve fluctuations and per-capita income fluctuations coincide, so $\bar{x}_N(t,\omega)$ and $\sigma_{x,N}(t)$ can be read as the bank's cash-reserve exposure and liquidity risk. This representation also suggests a practical panel regression for banks with granular account data. Let $b$ index banks, $s$ depositor segments, $g$ geographies or employer groups, and $t$ days or pay-cycle dates. A reduced-form implementation is
\begin{equation}\label{eq:PracticalBankRunRegression}
	\Delta CR_{bsgt}
	=
	\alpha+\beta_1\sigma_{sgt}
	+\beta_2\bigl(\sigma_{sgt}\times LA_{sgt}\bigr)
	+\Gamma^\prime Z_{bsgt}+\mu_b+\delta_s+\phi_g+\tau_t+\epsilon_{bsgt}.
\end{equation}
\noindent Here $\Delta CR_{bsgt}$ is the bank's change in cash reserves, net deposit outflow, abnormal withdrawal demand, or reserve shortfall for a bank-segment-pay-cycle cell. The liquidity-risk regressor $\sigma_{sgt}$ can be measured using account-balance volatility, deposit-flow volatility, or withdrawal-intensity volatility. The behavioral term $LA_{sgt}$ proxies loss aversion or paycheck stress with overdrafts, minimum-balance breaches, failed payments, balance drawdowns, call-center stress, or survey measures. The control vector $Z_{bsgt}$ should include income-shock variables such as payroll calendars, paycheck amounts, direct-deposit delays, employer concentration, layoffs, unemployment claims, and benefit-payment dates, together with bank-side controls such as reserve ratios, unused liquidity lines, deposit rates, branch and ATM access, digital-transfer limits, and relevant fixed effects. The coefficient $\beta_2$ is the practical bank-run-exposure statistic: it measures the incremental cash-reserve sensitivity generated when liquidity risk is amplified by loss aversion.

\vspace{0.5\baselineskip}
\begin{center}
	\red{\textbf{[Figure 2 goes about here: Simulated bank-run exposure regressions]}}
\end{center}
\vspace{0.5\baselineskip}

\tab The most important diagnostic is obtained by comparing the behavioral reserve-need estimate with the estimate from a liquidity-only regression that omits latent loss aversion. In the baseline simulation, the standard deviation of fitted reserve need is $0.238$ in the behavioral specification but only $0.030$ in the liquidity-only specification. Thus the omitted-loss-aversion model understates reserve-need volatility by about $87.3$ percent. \Cref{fig:OmittedLatentLossAversionReserveNeed} shows the mechanism: the liquidity-only fit makes reserve demand look smooth even though the latent behavioral model produces sharp reserve-need spikes in stress states.

\vspace{0.5\baselineskip}
\begin{center}
	\red{\textbf{[Figure 3 goes about here: Reserve-need stability when latent loss aversion is omitted]}}
\end{center}
\vspace{0.5\baselineskip}

\tab Under the maintained martingale-difference and previsibility conditions, the least-squares exposure estimate admits the following representation.\footnote{Under correct specification, regression residuals form a martingale difference sequence with respect to the information filtration \citep{Hamilton1994}. The regression coefficient arises as the linear projection coefficient that characterizes the predictable component in the Doob decomposition \citep{Durrett2019} of returns \citep{HansenRichard1987,Cochrane2005}.}
\begin{thm}[Bank run exposure martingale representation]\label{thm:BankRunExposureMartingale}
	Bank run exposure can be represented as
	\begin{equation}\label{eq:MainBetaItoProcess}
		\hat{\beta}_{x,N}(\bar{\lambda}(t),\omega)
		=
		E[\hat{\beta}_{x,N}(\bar{\lambda}(t),\omega)|\mathcal{F}_s]
		+\int^T_s b(u)dB(u,\omega),
	\end{equation}
	\noindent where
	\begin{equation}
		b(u)=\frac{\sigma_{x,N}(u)}{\int^T_0\sigma^2_{x,N}(v)dv}
	\end{equation}
	\noindent is a liquidity-risk factor.
\end{thm}
\tab Consequently, bank run exposure inherits the volatility of liquidity demand and the behavioral volatility of loss aversion. If exposure is separable in loss aversion, the maintained specification is that $\bar{\lambda}(t,\omega)$ has a half-Cauchy marginal law \citep{CharlesCadogan2018} with median parameter $m_\lambda\simeq 2.25$ \citep{TverKahn1992}. The parameter $m_\lambda$ is the continuous-time counterpart of the Hall--CRRA benchmark in \eqref{eq:LambdaHallHalfCauchyMultiplier}: in the simulations it is held fixed at the behavioral median, while in a bank-level implementation it can be set to the segment-specific value of $\lambda_t^{\mathrm{Hall}}$. Let
\begin{equation}\label{eq:HalfCauchyLossAversionMap}
	\bar{\lambda}(t)=h(Z_t)
	=
	m_\lambda\tan\biggl(\frac{\pi}{2}\Phi(Z_t)\biggr),
\end{equation}
\noindent where $\Phi$ is the standard normal distribution function and
\begin{equation}
	dZ_t=-\kappa Z_tdt+\sqrt{2\kappa}\,dB_\lambda(t)
\end{equation}
\noindent has stationary distribution $N(0,1)$ because it is an Ornstein-Uhlenbeck process with invariant Gaussian distribution \citep[Ch.~15]{KarlinTaylor1981,Oksendal2003}. By the probability integral transform \citep{Rosenblatt1952,Billingsley1995}, $\Phi(Z_t)$ is uniform on $(0,1)$ in stationarity. Applying the inverse half-Cauchy distribution function \citep{Devroye1986,JohnsonKotzBalakrishnan1994} implies that $\bar{\lambda}(t)$ is half-Cauchy with scale and median $m_\lambda$ in stationarity. Applying It\^o's formula to $h(Z_t)$ gives the loss-aversion stochastic differential equation (SDE)
\begin{equation}\label{eq:HalfCauchyLossAversionSDE}
	d\bar{\lambda}(t)
	=
	\bigl[-\kappa Z_t h^\prime(Z_t)+\kappa h^{\prime\prime}(Z_t)\bigr]dt
	+\sqrt{2\kappa}\,h^\prime(Z_t)dB_\lambda(t),
	\qquad Z_t=h^{-1}(\bar{\lambda}(t)).
\end{equation}
\noindent This specification replaces the Gaussian Ornstein-Uhlenbeck law in $h(Z_t)$ with a positive heavy-tailed process for $d\bar{\lambda}(t)$. It also supplies the missing dynamic link between loss aversion and bank-run exposure. Since the exposure coefficient in \eqref{eq:MainLinearIncomeRiskModel} is written as $\beta_{x,N}(\bar{\lambda}(t))$, a $C^2$ exposure map inherits dynamics from \eqref{eq:HalfCauchyLossAversionSDE}. By It\^o's formula,
\begin{equation}\label{eq:DynamicExposureCoefficientSDE}
 d\beta_{x,N}(\bar{\lambda}_t)
 =
 \beta_{x,N}^\prime(\bar{\lambda}_t)\,d\bar{\lambda}_t
 +
 \frac{1}{2}\beta_{x,N}^{\prime\prime}(\bar{\lambda}_t)\,d\langle\bar{\lambda}\rangle_t.
\end{equation}
\noindent Thus, the SDE for loss aversion is not an auxiliary embellishment; it is the mechanism that turns the static liquidity-risk slope into a dynamic bank-run exposure coefficient. The simulated paths in \cref{fig:LossAversionSDESimulation} show why the half-Cauchy specification is economically useful: most realizations remain near the behavioral benchmark, while rare draws generate large spikes in loss aversion. Those spikes are not treated as nuisance outliers in the model; they are precisely the stress states in which paycheck liquidity demand can be behaviorally amplified. The technical construction and weak-convergence argument are placed in the Internet Appendix.

\vspace{0.5\baselineskip}
\begin{center}
	\red{\textbf{[Figure 4 goes about here: Simulated half-Cauchy loss-aversion SDE]}}
\end{center}
\vspace{0.5\baselineskip}

\tab Because half-Cauchy exposure has a heavy right tail, level plots can obscure ordinary-state variation. The median log-exposure path is comparatively stable, but the pointwise $5$--$95$ percent band remains wide at many dates as shown in \cref{fig:LogBankRunExposureSDE}. Thus, the representative exposure state can look calm even though the cross-sectional range of bank-run exposure states is unstable at a given point in time.

\vspace{0.5\baselineskip}
\begin{center}
	\red{\textbf{[Figure 5 goes about here: Log bank-run exposure diagnostics]}}
\end{center}
\vspace{0.5\baselineskip}

\section{Proof-of-Concept Empirical Implementation}
\label{sec:EmpiricalImplementation}

\tab This section provides a proof-of-concept empirical implementation of the bank run exposure regression in \eqref{eq:PracticalBankRunRegression} using publicly available U.S.\ bank Call Report data. The objective is deliberately modest. Public Call Reports are quarterly bank-level regulatory filings, not account-level paycheck-cycle data. Thus, the exercise should be read as an external-validity and measurement demonstration: it shows how the exposure equation can be taken to real bank balance-sheet data, while making clear which variables would be measured more sharply with internal bank account data.

\subsection{Data}

\tab We use quarterly Consolidated Reports of Condition and Income (Call Reports) from the Federal Financial Institutions Examination Council (FFIEC) for U.S.\ commercial banks. The downloaded files cover 2001:Q1 through 2024:Q4. After constructing lagged deposit growth, rolling liquidity volatility, and the cash-reserve change, the analysis panel contains 8,909 bank-quarter observations for 232 banks from 2001:Q4 through 2024:Q4. The empirical sample is therefore useful for testing whether the exposure regression can be operationalized with public data, but it is not a substitute for the daily or pay-cycle account-level panel that the theory ultimately recommends.

\tab Because the sample includes the 2007--2009 financial crisis and the COVID-19 contraction, we also code recession-window indicators using the National Bureau of Economic Research (NBER) business-cycle dates.\footnote{See NBER,  \href{https://www.nber.org/research/data/us-business-cycle-expansions-and-contractions}{US Business Cycle Expansions and Contractions}. NBER dates the Great Recession from December 2007 to June 2009 and the COVID-19 recession from February 2020 to April 2020.} Since Call Reports are quarterly, the Great Recession dummy equals one for quarters overlapping December 2007--June 2009, i.e., 2007:Q4 through 2009:Q2. The COVID lockdown dummy equals one for quarters overlapping February--April 2020, i.e., 2020:Q1 and 2020:Q2. The Great Recession window contributes 799 observations in the analysis panel, while the COVID lockdown window contributes 146 observations.

\subsection{Variable Construction}

\tab \textbf{Cash-Reserve Change.}
Define bank-level reserve exposure as
\begin{equation}
	\Delta CR_{bt} = \Delta \left( \frac{\text{Cash}_{bt}}{\text{Total Assets}_{bt}} \right),
\end{equation}
\noindent where Cash is cash and balances due from depository institutions.

\tab \textbf{Total Deposits.}
The deposit denominator is constructed defensively across reporting forms. We use consolidated total deposits when available and otherwise reconstruct total deposits as domestic plus foreign deposits:
\begin{equation}
	\text{Total Deposits}_{bt}
	=
	\begin{cases}
		\text{RCFD2200}_{bt}, & \text{if } \text{RCFD2200}_{bt}\text{ is reported},\\
		\text{RCON2200}_{bt}+\text{RCFN2200}_{bt}, & \text{otherwise}.
	\end{cases}
	\label{eq:EmpiricalTotalDeposits}
\end{equation}
\noindent In the FFIEC public bulk extract used here, all observations in the cleaned panel use the fallback construction $\text{RCON2200}+\text{RCFN2200}$. The cleaning code then checks for negative deposits, missing deposits, zero deposits with positive assets, and mechanical share violations. The resulting master panel has no negative or missing total deposits. There are 2,146 zero-deposit observations with positive assets; these observations are retained in the master file for transparency, but deposit-growth calculations are set to missing whenever the lagged deposit denominator is nonpositive.

\vspace{0.5\baselineskip}
\begin{center}
	\red{\textbf{[Table 1 goes about here: Public Call Report panel diagnostics]}}
\end{center}
\vspace{0.5\baselineskip}

\tab \textbf{Liquidity Risk.}
Liquidity volatility is measured as
\begin{equation}
	\sigma_{bt}
	=
	\operatorname{sd}_{4q}
	\left(
	\frac{\text{Total Deposits}_{bt}-\text{Total Deposits}_{b,t-1}}
	{\text{Total Deposits}_{b,t-1}}
	\right),
\end{equation}
\noindent where the rolling standard deviation is computed over four quarters. Deposit growth and $\sigma_{bt}$ are winsorized by quarter at the 1 percent tails to reduce the influence of very small-bank denominator effects.

\tab \textbf{Loss-Aversion Proxy.}
Consistent with the paycheck-to-paycheck model, we proxy depositor behavioral exposure using the retail funding share:
\begin{equation}
	LA_{bt}
	=
	\text{Retail Share}_{bt}
	=
	1-\frac{\text{Brokered Deposits}_{bt}}{\text{Total Deposits}_{bt}}.
\end{equation}
\noindent Missing brokered deposits are treated as zero for this proxy, since many small banks do not report brokered deposits separately. Retail-share observations outside $[0,1]$ are set to missing. This proxy is admittedly coarse: retail funding is not loss aversion itself. Rather, it is a public-data proxy for clientele exposure to paycheck-sensitive funding. The conceptual link is consistent with evidence that income timing affects household expenditure and liquidity behavior \citep{Stephens2006,JappelliPistaferri2010}.

\tab \textbf{Composite Loss-Aversion Proxy.}
To move beyond a single retail-share proxy, we also construct a regression-weighted composite exposure index. Let
\begin{align}
	\text{Transaction Share}_{bt}
	&=
	\frac{\text{Noninterest Demand Deposits}_{bt}
		+\text{Interest-Bearing Transaction Accounts}_{bt}}
	{\text{Total Deposits}_{bt}},\\
	\text{Core Share}_{bt}
	&=
	\frac{\text{Total Deposits}_{bt}
		-\text{Brokered Deposits}_{bt}
		-\text{Large Time Deposits}_{bt}}
	{\text{Total Deposits}_{bt}},\\
	\text{Consumer Loan Share}_{bt}
	&=
	\frac{\text{Consumer Loans}_{bt}}{\text{Total Loans}_{bt}}.
\end{align}
\noindent The large-time-deposit item is harmonized across the reporting threshold change by using large time deposits of at least \$100,000 through 2009 and large time deposits of at least \$250,000 thereafter. The three component shares are winsorized and standardized to obtain $TxnZ_{bt}$, $CoreZ_{bt}$, and $ConsLoanZ_{bt}$.

\tab The regression-weighted index is constructed in two steps. First, each standardized component is interacted separately with liquidity volatility in a bank and quarter fixed-effect regression. Second, the absolute values of the three interaction coefficients are normalized to sum to one:
\begin{equation}
	w_k
	=
	\frac{|\widehat{\beta}_{2k}|}
	{\sum_j|\widehat{\beta}_{2j}|},
	\qquad
	k\in\{Txn,Core,ConsLoan\}.
\end{equation}
\noindent In the present sample, the normalized weights are $0.169$ for transaction share, $0.537$ for core share, and $0.294$ for consumer-loan share. We then define
\begin{equation}
	LA^{reg}_{bt}
	=
	\operatorname{scale}
	\left(
	0.169\,TxnZ_{bt}
	+
	0.537\,CoreZ_{bt}
	+
	0.294\,ConsLoanZ_{bt}
	\right).
\end{equation}
\noindent For robustness, we also estimate the same regressions with an equal-weight index,
\begin{equation}
	LA^{equal}_{bt}
	=
	\operatorname{scale}
	\left(
	\frac{TxnZ_{bt}+CoreZ_{bt}+ConsLoanZ_{bt}}{3}
	\right).
\end{equation}
\noindent These indices remain public-data proxies. They should be interpreted as balance-sheet measures of clientele and product exposure that may correlate with loss-averse liquidity demand, not as direct measures of depositor psychology.

\tab \textbf{Uninsured Deposits.}
We also construct
\begin{equation}
	\text{Uninsured Share}_{bt}
	=
	\frac{\text{Uninsured Deposits}_{bt}}{\text{Total Deposits}_{bt}},
\end{equation}
\noindent and set values above one to missing. Uninsured-share coverage is incomplete in the public bulk files and should be interpreted as a balance-sheet control rather than as an account-level run trigger. This is an important caveat because recent evidence on the 2023 episode shows that uninsured deposits are central to cross-bank stress \citep{ChoiGoldsmithPinkhamYorulmazer2023}.

\tab \textbf{Small Banks.}
To test whether the public-data exposure relation differs for small banks, define
\begin{equation}
	SmallBank_{bt}
	=
	\mathds{1}
	\left\{
	\text{Total Assets}_{bt}<\$1\text{ billion}
	\;\text{and}\;
	\text{Only Domestic Offices}_{bt}=1
	\right\}.
\end{equation}
\noindent Call Report balance-sheet variables are reported in thousands of dollars, so the asset cutoff is implemented as $\text{Total Assets}_{bt}<1{,}000{,}000$. In the public bulk files, the domestic-office screen is proxied by the absence of foreign-office deposits, i.e., $\text{RCFN2200}_{bt}=0$ or missing. The analysis panel contains 796 small-bank observations, equal to 8.9\% of the bank-quarter sample.

\subsection{Empirical Specification}

\tab The estimated panel regression is
\begin{equation}
	\Delta CR_{bt}
	=
	\alpha
	+
	\beta_1 \sigma_{bt}
	+
	\beta_2 (\sigma_{bt} \times LA_{bt})
	+
	\mu_b
	+
	\tau_t
	+
	\varepsilon_{bt},
\end{equation}
\noindent where $\mu_b$ are bank fixed effects and $\tau_t$ are time fixed effects. Standard errors are clustered two-way by bank and quarter following \citet{CameronGelbachMiller2011}. We estimate this specification first with $LA_{bt}=\text{Retail Share}_{bt}$ and then with $LA_{bt}=LA^{reg}_{bt}$ and $LA_{bt}=LA^{equal}_{bt}$. Thus, the empirical exercise asks whether a regression-weighted composite behavioral-exposure proxy has more explanatory content than retail share alone.

\tab To account for aggregate stress periods, we also estimate recession-window specifications. First, we estimate a bank-fixed-effect specification with Great Recession and COVID lockdown dummies. Second, because pure recession dummies are absorbed by quarter fixed effects, we estimate quarter-fixed-effect interaction specifications in which the recession-window indicators interact with $\sigma_{bt}$, $LA_{bt}$, and $\sigma_{bt}\times LA_{bt}$. The recession-window specifications are run for both the retail-share proxy and the regression-weighted composite proxy.

\subsection{Results}

\tab\cref{tab:PublicCallReportExposureRegression} reports a compact version of the retail-share fixed-effect estimates. The baseline interaction between liquidity volatility and retail share is positive but statistically imprecise in the public quarterly data. That result should not be overinterpreted as a rejection of the behavioral mechanism, because retail share is only a coarse public proxy for loss aversion and paycheck stress. The specification that adds uninsured-share exposure shows a negative and statistically significant interaction between liquidity volatility and uninsured share. One interpretation is that banks with greater uninsured-deposit exposure adjust cash ratios differently once uninsured funding risk is accounted for directly. The SVB-period triple interaction is positive and marginally significant in the retail-share specification without the uninsured-share interaction, which is consistent with heightened behavioral amplification during a salient run-risk episode, but the precision is weak.

\vspace{0.5\baselineskip}
\begin{center}
	\red{\textbf{[Table 2 goes about here: Proof-of-concept bank run exposure regressions]}}
\end{center}
\vspace{0.5\baselineskip}

\tab\cref{tab:LAProxyComparison} compares retail share with the regression-weighted composite proxy $LA^{reg}_{bt}$ and the equal-weight proxy $LA^{equal}_{bt}$. The key interaction remains imprecise across all proxy choices, which is unsurprising given that the data are quarterly public balance-sheet variables rather than account-level paycheck-cycle observations. Nevertheless, the composite proxy has modestly more explanatory content than retail share: the baseline $R^2$ rises from $0.043$ with retail share to $0.049$ with $LA^{reg}_{bt}$, and the uninsured-control specification rises from $0.057$ to $0.061$. The equal-weight index produces nearly the same fit as the regression-weighted index, suggesting that the improvement is not driven by an unstable weighting scheme. Thus, the result is best interpreted as a measurement lesson: balance-sheet information that spans transaction deposits, core funding, and consumer-loan exposure appears to summarize public-data behavioral exposure slightly better than retail share alone, but the public data remain too coarse to estimate the theory's account-level mechanism sharply.

\vspace{0.5\baselineskip}
\begin{center}
	\red{\textbf{[Table 3 goes about here: Retail share and composite loss-aversion proxy comparison]}}
\end{center}
\vspace{0.5\baselineskip}

\tab\cref{tab:LAProxySVBComparison} extends the same proxy-comparison logic to the Post-SVB event window. The retail-share event model reproduces the positive and marginally significant triple interaction, $\sigma_{bt}\times \text{Retail Share}_{bt}\times PostSVB_t=0.141$ with standard error $0.084$. The corresponding $LA^{reg}_{bt}$ triple interaction is positive but smaller and statistically imprecise, $0.009$ with standard error $0.008$; with the uninsured PostSVB interaction added, it is $0.011$ with standard error $0.009$. The equal-weight composite produces still smaller imprecise Post-SVB triple interactions. Thus, the SVB event-window exercise supports a careful interpretation: the retail-share proxy captures a marginal post-SVB amplification signal in public data, while the composite proxies improve overall fit but do not produce a sharper Post-SVB amplification coefficient.

\vspace{0.5\baselineskip}
\begin{center}
	\red{\textbf{[Table 4 goes about here: Post-SVB event-window comparison across exposure proxies]}}
\end{center}
\vspace{0.5\baselineskip}

\tab\cref{tab:SmallBankInteractions} reports small-bank interaction estimates. The small-bank dummy equals one for banks with less than \$1 billion in total assets and only domestic offices. In the retail-share specification, the small-bank interaction with the behavioral exposure slope is positive. Without the uninsured-share control, the coefficient on $\sigma_{bt}\times \text{Retail Share}_{bt}\times SmallBank_{bt}$ is $0.058$ with standard error $0.096$; with uninsured share included, it rises to $0.076$ with standard error $0.011$. The regression-weighted composite proxy gives a smaller but still positive small-bank interaction, $0.010$ with standard error $0.006$, in the uninsured-control specification. The equal-weight composite yields $0.014$ with standard error $0.005$ in the same specification. Thus, the public-data evidence suggests that small-bank status may matter for exposure measurement, but the inference depends on the proxy and control set. The result is best read as a robustness and heterogeneity check: small banks can be added to the monitoring regression through interaction terms, but the public Call Report data do not identify the account-level behavioral mechanism by themselves.

\vspace{0.5\baselineskip}
\begin{center}
	\red{\textbf{[Table 5 goes about here: Small-bank interaction regressions]}}
\end{center}
\vspace{0.5\baselineskip}

\tab We also estimated Post-SVB specifications that interact the exposure slope with $SmallBank_{bt}$. The composite-proxy four-way interactions, $\sigma_{bt}\times LA^{reg}_{bt}\times PostSVB_t\times SmallBank_{bt}$ and $\sigma_{bt}\times LA^{equal}_{bt}\times PostSVB_t\times SmallBank_{bt}$, are small and statistically imprecise. The retail-share four-way specification produces large offsetting coefficients on the small-bank Post-SVB terms, which is consistent with near-collinearity created by combining a narrow event window, a high retail-share proxy, and a small number of small-bank event observations. We therefore treat the Post-SVB small-bank estimates as a diagnostic robustness check rather than as a separate structural finding.

\tab The recession-window estimates are reported in\cref{tab:PublicCallReportMacroRegression} for the retail-share proxy, with the corresponding composite-proxy estimates generated in the empirical output. In the bank-fixed-effect specification without quarter fixed effects, the Great Recession dummy is economically small and statistically insignificant. The COVID lockdown dummy is positive and marginally significant, but its magnitude is small. Once quarter fixed effects are included, the recession-window dummies themselves are absorbed, as expected; identification comes from the interaction of recession windows with bank-level liquidity risk and the chosen exposure proxy. The retail-share Great Recession interaction terms are not statistically significant. The COVID lockdown retail-share interaction is positive and marginally significant, but the COVID interaction with $\sigma_{bt}\times LA_{bt}$ is not statistically significant. The composite-proxy estimates tell the same story: the $LA^{reg}_{bt}$ Great Recession triple interaction is $-0.005$ with standard error $0.013$, and the COVID triple interaction is $-0.002$ with standard error $0.005$. The $LA^{reg}_{bt}\times COVID_t$ coefficient is marginally positive, but that term is a level-shift in the proxy rather than the exposure-slope coefficient. Thus, the public Call Report exercise does not show that the Great Recession or the short COVID recession materially changes the estimated behavioral exposure slope. Rather, it shows that the regression can be stress-tested against major aggregate episodes without mechanically attributing all exposure variation to those episodes.

\vspace{0.5\baselineskip}
\begin{center}
	\red{\textbf{[Table 6 goes about here: Recession-window robustness checks]}}
\end{center}
\vspace{0.5\baselineskip}

\tab The lesson from the public-data exercise is therefore methodological rather than definitive. The regression can be implemented, the denominator and share variables can be audited, the exposure coefficients can be stress-tested against major macro episodes, and the estimated coefficients can be used as monitoring statistics. But the most theory-consistent implementation would use internal bank data: account-balance volatility, paycheck calendar information, direct-deposit delays, overdraft events, failed payments, digital-transfer intensity, and depositor-segment identifiers. Public Call Reports are too coarse to observe the within-quarter behavioral mechanism directly.

\tab This caveat is central for identification. If total deposits are mismeasured, the interaction $\sigma_{bt}\times LA_{bt}$ can be mechanically inflated because both liquidity volatility and the retail-share proxy inherit denominator error. The cleaning procedure in \eqref{eq:EmpiricalTotalDeposits} and\cref{tab:PublicCallReportDiagnostics} is therefore not a housekeeping detail; it is part of the identification discipline needed to estimate bank run exposure credibly.

\section{Conclusion}\label{sec:Conclusion}
\tab This paper develops a behavioral theory of bank-run exposure for a paycheck-to-paycheck economy in which depositors receive income through demand deposits and are loss averse over declines in income and consumption. The central mechanism is simple: when liquidity risk interacts with state-dependent loss aversion, ordinary cash-demand fluctuations can become behaviorally amplified withdrawal pressure. Market clearing links those fluctuations to bank cash reserves, so bank-run exposure is not only a balance-sheet object; it is also a clientele-specific behavioral exposure.

\tab The model makes three formal contributions. First, it identifies stress states in which sufficiently high subjective bad-state probabilities and loss aversion support run exposure. These states define the Bank Run Exposure State Space. Second, the stopped-process representation links suspension of convertibility to the withdrawal order, the reserve constraint, and depositor beliefs about bad income states. Third, the martingale and stochastic-loss-aversion representation show how a static liquidity-risk slope becomes a dynamic exposure coefficient. In particular, the half-Cauchy loss-aversion process preserves ordinary-state behavior while allowing rare fear-state spikes, which are precisely the states in which reserve demand can become unstable.

\tab The public Call Report exercise is best read as a proof of concept rather than a definitive empirical test. With 232 banks and 8,909 bank-quarter observations from 2001:Q4 through 2024:Q4, the proposed exposure regression can be implemented, audited, and stress-tested. Retail-share exposure gives the expected positive but imprecise baseline interaction. A regression-weighted composite proxy based on transaction-deposit share, core-deposit share, and consumer-loan share modestly improves fit relative to retail share, and small-bank and Post-SVB interactions provide suggestive but proxy-sensitive evidence of amplification. The main empirical lesson is therefore measurement discipline: public quarterly balance-sheet data can illustrate the exposure regression, but the theory calls for account-level pay-cycle data, including deposit timing, overdrafts, failed payments, balance drawdowns, digital-transfer intensity, and depositor-segment identifiers.

\tab The practical implication is that banks and supervisors can treat bank-run exposure as an estimable monitoring statistic. Regressions of reserve changes, abnormal outflows, or reserve shortfalls on liquidity volatility and its interaction with loss-aversion proxies can reveal whether cash-reserve demand is being amplified by paycheck stress. A liquidity-only model can make reserve needs look too stable, understating the volatility of cash demand before suspension or failure is observed. More broadly, Diamond--Dybvig type consumption-based approaches to bank runs should be extended to include loss aversion to declines in consumption and income. The accompanying Internet Appendix collects the technical derivations, simulation details, and diagnostic output so that the main paper remains focused on the economic mechanism.

\section{Proofs}\label{sec:Proofs}

\subsection{Proofs for \cref{subsec:LiquidityPreferenceTheory}: Liquidity Preference and Loss Aversion}\label{subsec:ProofsLiquidityPreference}

\begin{proof}[Derivation of \eqref{eq:TobinProspectSlope}]
	Following \citet[pg.~75]{Tobin1958}, interpret $V(\mu_x,\sigma_x)$ as defining a local indifference contour in mean-risk space. Along a constant expected value contour,
	\begin{equation}\label{eq:TotalDiffExpectedProspectValue}
		dV(\mu_x,\sigma_x)=0.
	\end{equation}
	\noindent Since
	\begin{equation}\label{eq:StandardizedWealthChange}
		x(\mathfrak{z})=\mu_x+\mathfrak{z}\sigma_x,
	\end{equation}
	\noindent we have
	\begin{equation}\label{eq:dxStandardizedWealthChange}
		dx(\mathfrak{z})=d\mu_x+\mathfrak{z}\,d\sigma_x.
	\end{equation}
	\noindent Under the maintained assumption that $w(d\Phi(\mathfrak{z}))$ is fixed with respect to local changes in $(\mu_x,\sigma_x)$, differentiation under the integral gives
	\begin{align}
		dV(\mu_x,\sigma_x)
		&=
		\int^\infty_{-\infty}
		v^\prime(x(\mathfrak{z}))
		\bigl(d\mu_x+\mathfrak{z}\,d\sigma_x\bigr)
		w(d\Phi(\mathfrak{z})) \notag\\
		&=
		\left[
		\int^\infty_{-\infty}
		v^\prime(x(\mathfrak{z}))w(d\Phi(\mathfrak{z}))
		\right]d\mu_x
		+
		\left[
		\int^\infty_{-\infty}
		\mathfrak{z}v^\prime(x(\mathfrak{z}))w(d\Phi(\mathfrak{z}))
		\right]d\sigma_x.
		\label{eq:TotalDiffExpandedExpectedProspectValue}
	\end{align}
	\noindent Setting \eqref{eq:TotalDiffExpandedExpectedProspectValue} equal to zero and solving for $d\mu_x/d\sigma_x$ gives \eqref{eq:TobinProspectSlope}.
\end{proof}

\begin{proof}[Proof of \Cref{lem:SignOfIncomeIncomeRiskTradeoff}]
	Let
	\begin{align}
		g(\mu_x,\sigma_x;\lambda) &= \int^\infty_{-\infty}v^\prime(\mu_x+\mathfrak{z}\sigma_x)\;w(d\Phi(\mathfrak{z}))\\
		&= \int_{\{x(\mathfrak{z})\geq 0\}} v_g^\prime(x(\mathfrak{z})) w(d\Phi(\mathfrak{z}))+\lambda \int_{\{x(\mathfrak{z})<0\}} v_\ell^\prime (-x(\mathfrak{z}))w(d\Phi(\mathfrak{z}))\label{eq:AffineG}\\
		\intertext{and, for a given $\lambda$, let}
		\beta(x;\lambda) &=\frac{d\mu_x}{d\sigma_x}\label{eq:BetaSlope}\\
		\intertext{be the local slope in $(\mu_x,\sigma_x)$ space. For fixed gain/loss regions, the local derivative of the slope is the derivative of \eqref{eq:TobinProspectSlope}, and can be written as}
		\frac{\partial\beta(x;\lambda)}{\partial x} &=\frac{g(\cdot)\frac{\partial}{\partial x}[I_1+\lambda I_2]-(I_1+\lambda I_2)\frac{\partial}{\partial x}g(\cdot)}{g^2(\cdot)}\\
		\intertext{where}
		I_1 &= \int_{\{x(\mathfrak{z})\geq 0\}} \mathfrak{z}v_g^\prime(x(\mathfrak{z})) w(d\Phi(\mathfrak{z}))\\
		I_2 &= \int_{\{x(\mathfrak{z})<0\}} \mathfrak{z}v_\ell^\prime (-x(\mathfrak{z}))w(d\Phi(\mathfrak{z}))\\
		\intertext{and}
		J_1 &= -\int_{\{x(\mathfrak{z})\geq 0\}} v_g^{\prime\prime}(x(\mathfrak{z})) w(d\Phi(\mathfrak{z})) > 0 \label{eq:J1Def}\\
		J_2 &= -\int_{\{x(\mathfrak{z})<0\}} v_\ell^{\prime\prime} (-x(\mathfrak{z}))w(d\Phi(\mathfrak{z})) > 0. \label{eq:J2Def}
	\end{align}
	\noindent Thus $J_1$ and $J_2$ collect the absolute curvature terms in the gain and loss regions. On a fixed gain/loss partition,
	\begin{align}
		\frac{\partial}{\partial x}g(\cdot) &= \frac{\partial}{\partial x}\biggl[\int_{\{x(\mathfrak{z})\geq 0\}} v_g^\prime(x(\mathfrak{z})) w(d\Phi(\mathfrak{z}))+\lambda \int_{\{x(\mathfrak{z})<0\}} v_\ell^\prime (-x(\mathfrak{z}))w(d\Phi(\mathfrak{z}))\biggr]\\
		&= \biggl[\int_{\{x(\mathfrak{z})\geq 0\}} v_g^{\prime\prime}(x(\mathfrak{z})) w(d\Phi(\mathfrak{z}))-\lambda \int_{\{x(\mathfrak{z})<0\}} v_\ell^{\prime\prime} (-x(\mathfrak{z}))w(d\Phi(\mathfrak{z}))\biggr]\\
		&= -J_1+\lambda J_2.
	\end{align}
	\noindent Hence the sign of $\beta^\prime(x;\lambda)$ is local and is determined by the numerator
	\begin{align}
		\beta^\prime(x;\lambda) &= \frac{\psi(\lambda)}{g^2(\mu_x,\sigma_x;\lambda)}\label{eq:BetaSlopeDirection}\\
		\intertext{where}
		\psi(\lambda)&=g(\cdot)\biggl(\frac{\partial I_1}{\partial x}+\lambda\frac{\partial I_2}{\partial x}\biggr)-(I_1+\lambda I_2)(-J_1+\lambda J_2).\label{eq:PsiLambda}
	\end{align}
	\noindent Since $g$ is affine in $\lambda$ on a fixed gain/loss partition in \eqref{eq:AffineG}, this expression gives a local sign condition for determining whether higher loss aversion steepens or reverses the mean-risk tradeoff at a given point $x=x_0$.
\end{proof}

\begin{proof}[Proof of \Cref{prop:LiquidityPrefEquation}]
	Equation \eqref{eq:TobinProspectSlope} defines the compensating mean-risk slope along a local indifference contour. Writing that slope as $\beta(x;\lambda)$ gives $d\mu_x=\beta(x;\lambda)d\sigma_x$. For a local linear approximation around the relevant point in mean-risk space, integrating this differential relation gives $\mu_x\approx \beta(x;\lambda)\sigma_x$, up to the local intercept absorbed by the reference point.
\end{proof}

\begin{proof}[Proof of \Cref{cor:RiskSeekingOverIncomeLoss}]
	By \Cref{lem:SignOfIncomeIncomeRiskTradeoff}, the local sign of $\beta^\prime(x;\lambda)$ is the sign of $\psi(\lambda)$ whenever
	\[
	g^2(\mu_x,\sigma_x;\lambda)\neq 0.
	\]
	\noindent If $\lambda_c=\inf\{\lambda>0:\psi(\lambda)<0\}$ and $\psi(\lambda)$ is locally monotone decreasing in the relevant region, then $\psi(\lambda)<0$ for $\lambda>\lambda_c$. Hence $\beta^\prime(x;\lambda)<0$ for $\lambda>\lambda_c$.
\end{proof}

\subsection{Proofs for \cref{subsec:MyopicLossAversion}: Myopic Loss Aversion Paycheck to Paycheck}\label{subsec:ProofsMyopicLossAversion}

\begin{proof}[Proof of \Cref{thm:LossAversionBasedMPC}]
	Suppose that the marginal propensity to consume is deterministic, that income is subject to state dependent transitory shocks $\eta_{i\omega}$, and that permanent income $Y_p$ is the same in each of the two time periods. At the start of period $1$, agents are uncertain about transitory shocks to permanent income, so
	\begin{align}
		C_{1\omega}&= k(\rho,d)Y_p + \eta_{1\omega}
		\intertext{However, according to \cref{assum:PermanentConsumeFluctuationHypothesis}, there exists $\tilde{k}(\rho,d;\lambda)$ such that}
		C_{1\omega}&=\tilde{k}(\rho,d;\lambda(\omega))Y_p,\quad\text{whereupon}\label{eq:MarginalpropenConsumeWithLambdaYp}\\
		\tilde{k}(\rho,d;\lambda(\omega))Y_p&=k(\rho,d)Y_p + \eta_{1\omega},\quad\text{and}\label{eq:MarginalPropenConsumeWithLambdaError}\\
		\tilde{k}(\rho,d;\lambda(\omega))&=k(\rho,d)+\tilde{\eta}_{1\omega},\quad\text{where}\label{eq:MarginalPropenConsumeWithLambda}\\
		\tilde{\eta}_{1\omega}&=\frac{\eta_{1\omega}}{Y_p}.
	\end{align}
	\noindent If $\tilde{k}(\rho,d;\lambda)$ is separable in $\lambda$, so that $\tilde{k}(\rho,d;\lambda)=\tilde{k}(\rho,d)\lambda$, then \eqref{eq:MarginalPropenConsumeWithLambdaError} gives
	\begin{align}
		\lambda(\rho,d,\omega)&=\biggl(1+\frac{\tilde{\eta}_{1\omega}}{k(\rho,d)}\biggr)\frac{k(\rho,d)}{\tilde{k}(\rho,d)},\\
		\tilde{\eta}_{1\omega}&=\tilde{k}(\rho,d)\lambda(\omega)-k(\rho,d).
	\end{align}
	\noindent From \eqref{eq:MarginalPropenConsumeWithLambda} and \eqref{eq:MarginalpropenConsumeWithLambdaYp}, the one-period consumption ratchet is
	\begin{align}
		c_{1\omega}&=\tilde{\eta}_{1\omega}Y_p\\
		&=(\tilde{k}(\rho,d)\lambda(\omega)-k(\rho,d))Y_p\\
		&=(\tilde{k}(\rho,d;\lambda(\omega))-k(\rho,d))Y_p.
	\end{align}
	\noindent Provided $Y_p>0$, the condition $c_{1,\omega}\geq 0$ is equivalent to $\tilde{k}(\rho,d;\lambda(\omega))\geq k(\rho,d)$.
\end{proof}

\subsection{Proofs for \cref{subsec:OptimalConsumptionProblem}: The Optimal Consumption Problem}\label{subsec:ProofsOptimalConsumption}

\begin{proof}[Derivation of \Cref{lem:BankRunTriggerEndogenLossAver}]
	The Lagrangian over the state space is
	\begin{equation}
		\begin{split}
			\mathcal{L} &= v(c_0)+\rho\Sigma_{\omega\in\Omega}\pi_\omega v(c_{1\omega}) -\mu_0(c_0-x_0)\\
			&\quad -\Sigma_{\omega\in\Omega}\mu_{1\omega}\pi_\omega (c_{1\omega}-x_{1\omega}),
		\end{split}
	\end{equation}
	\noindent where $\mu_0$ is unrestricted in sign and $\mu_{1\omega}\geq 0$. The necessary first-order conditions for an interior maximum are
	\begin{align}
		v^\prime(c_0)-\mu_0 &= 0\label{eq:FOC_c0}\\
		\rho\Sigma_{\omega\in\Omega}\pi_\omega v^\prime(c_{1\omega})-\Sigma_{\omega\in\Omega}\mu_{1\omega}\pi_\omega &= 0.\label{eq:FOC_c1_avg}
	\end{align}
	\noindent The statewise first-order conditions are $\rho v^\prime(c_{1\omega}) = \mu_{1\omega}$ for each $\omega$; averaging across states gives \eqref{eq:FOC_c1_avg}. Let
	\begin{equation}
		\bar{\mu}_1 =\Sigma_{\omega\in\Omega}\mu_{1\omega}\pi_\omega.
	\end{equation}
	\noindent According to the Karush-Kuhn-Tucker conditions \citep[see, e.g.,][pg.~42]{LuenbergerYe2008}, slack constraints have zero multipliers, while binding constraints have positive multipliers. The weighted first-order condition does not imply that each statewise marginal term is zero. Instead, in a loss state $\omega^\prime$, the weighted marginal loss dominates the weighted marginal gain if
	\begin{equation}\label{eq:StressStateTrigger}
		\mathbb{I}_{\{c_{1\omega^\prime}<0\}}\pi^-_{\omega^\prime}\lambda v^\prime_\ell(-c_{1\omega^\prime})
		>
		\pi^+_{\omega^\prime}v^\prime_g(c_{1\omega^\prime}).
	\end{equation}
	\noindent Equivalently, for such a state,
	\begin{equation}\label{eq:LossAverIndexLowerBound}
		\lambda(\omega^\prime)>\frac{\pi^+_{\omega^\prime}v^\prime_g(c_{1\omega^\prime})}{\pi^-_{\omega^\prime}v^\prime_\ell(-c_{1\omega^\prime})}.
	\end{equation}
	\noindent The indicator $\mathbb{I}_{\{c_{1\omega^\prime}<0\}}$ is canceled from both sides because $\pi^-_{\omega^\prime}>0$ in a loss state. This is the sufficient stress-state condition stated in \Cref{lem:BankRunTriggerEndogenLossAver}.
\end{proof}

\subsection{Proofs for \cref{sec:BankRunTopology}: Bank Run Exposure State Space}\label{subsec:ProofsStateSpace}

\begin{proof}[Proof of \Cref{thm:BankRunTopology}]
	Use \Cref{lem:AdmissBankRunStates} and \Cref{lem:TabooStatesLossAver} as in \eqref{eq:OmegaRun} to define
	\begin{equation*}
		\underline{\Omega}^{\text{run}} = \underline{\Omega}^{\text{admiss}}\cap\underline{\Omega}^{\text{taboo}}
	\end{equation*}
	\noindent and endow it with the pullback topology $\mathcal{T}^{\text{run}}=\{\lambda^{-1}(O):O\in\mathcal{T}_\Lambda\}$. The non-emptiness follows from \Cref{lem:AdmissBankRunStates} and the existence of taboo states established in \Cref{lem:TabooStatesLossAver}.
\end{proof}

\clearpage
\section{Tables}

\begin{table}[htbp]
	\centering
	\caption{Public Call Report panel diagnostics}
	\label{tab:PublicCallReportDiagnostics}
	\begin{tabular}{lc}
		\hline
		Diagnostic & Value \\
		\hline
		Analysis observations & 8,909 \\
		Banks & 232 \\
		Mean retail share & 0.932 \\
		Mean uninsured share & 0.368 \\
		Mean transaction share & 0.145 \\
		Mean core share & 0.932 \\
		Mean consumer-loan share & 0.008 \\
		Small-bank observations & 796 \\
		Small-bank share & 0.089 \\
		Mean liquidity volatility $\sigma_{bt}$ & 0.058 \\
		Negative total deposits & 0 \\
		Missing total deposits & 0 \\
		Raw retail-share violations & 3 \\
		Raw uninsured-share violations & 30 \\
		$\sigma_{bt}>1$ after winsorization & 0 \\
		\hline
	\end{tabular}
\end{table}

\begin{table}[htbp]
	\centering
	\caption{Proof-of-concept bank run exposure regressions}
	\label{tab:PublicCallReportExposureRegression}
	\begin{tabular}{lcccc}
		\hline
		& Baseline & Uninsured & SVB & SVB + Uninsured \\
		\hline
		$\sigma_{bt}$ & -0.011 & 0.005 & -0.019 & -0.010 \\
		& (0.014) & (0.012) & (0.021) & (0.017) \\
		$\sigma_{bt}\times LA_{bt}$ & 0.008 & -0.003 & 0.018 & 0.008 \\
		& (0.013) & (0.012) & (0.021) & (0.017) \\
		$\sigma_{bt}\times\text{Uninsured Share}_{bt}$ &  & -0.014$^{*}$ &  &  \\
		&  & (0.007) &  &  \\
		$\sigma_{bt}\times LA_{bt}\times PostSVB_t$ &  &  & 0.141$^{+}$ & 0.138 \\
		&  &  & (0.084) & (0.100) \\
		Observations & 8,909 & 7,766 & 8,909 & 7,766 \\
		Bank fixed effects & Yes & Yes & Yes & Yes \\
		Quarter fixed effects & Yes & Yes & Yes & Yes \\
		\hline
	\end{tabular}
	\begin{flushleft}
		\footnotesize Notes: The dependent variable is the quarterly change in the cash-to-assets ratio. Standard errors are clustered two-way by bank and quarter. $^{+}p<0.10$, $^{*}p<0.05$.
	\end{flushleft}
\end{table}

\begin{table}[htbp]
	\centering
	\caption{Retail share and composite loss-aversion proxy comparison}
	\label{tab:LAProxyComparison}
	\begin{tabular}{lcccccc}
		\hline
		& Retail & $LA^{reg}$ & $LA^{equal}$ & Retail & $LA^{reg}$ & $LA^{equal}$ \\
		& Share &  &  & + Unins. & + Unins. & + Unins. \\
		\hline
		$\sigma_{bt}$ & -0.011 & -0.003 & -0.003 & 0.005 & -0.000 & -0.000 \\
		& (0.014) & (0.003) & (0.003) & (0.012) & (0.004) & (0.004) \\
		$\sigma_{bt}\times \text{Retail Share}_{bt}$ & 0.008 &  &  & -0.003 &  &  \\
		& (0.013) &  &  & (0.012) &  &  \\
		$\sigma_{bt}\times LA^{reg}_{bt}$ &  & 0.001 &  &  & -0.000 &  \\
		&  & (0.002) &  &  & (0.001) &  \\
		$\sigma_{bt}\times LA^{equal}_{bt}$ &  &  & 0.000 &  &  & -0.000 \\
		&  &  & (0.003) &  &  & (0.002) \\
		$\sigma_{bt}\times\text{Uninsured Share}_{bt}$ &  &  &  & -0.014$^{*}$ & -0.010 & -0.010 \\
		&  &  &  & (0.007) & (0.008) & (0.008) \\
		Observations & 8,909 & 8,712 & 8,712 & 7,766 & 7,643 & 7,643 \\
		$R^2$ & 0.043 & 0.049 & 0.049 & 0.057 & 0.061 & 0.061 \\
		Within $R^2$ & 0.000 & 0.000 & 0.000 & 0.000 & 0.000 & 0.000 \\
		Bank fixed effects & Yes & Yes & Yes & Yes & Yes & Yes \\
		Quarter fixed effects & Yes & Yes & Yes & Yes & Yes & Yes \\
		\hline
	\end{tabular}
	\begin{flushleft}
		\footnotesize Notes: The dependent variable is the quarterly change in the cash-to-assets ratio. Standard errors are clustered two-way by bank and quarter. $LA^{reg}_{bt}$ is the regression-weighted composite index based on transaction share, core share, and consumer-loan share. $LA^{equal}_{bt}$ uses equal weights on the same standardized components. $^{*}p<0.05$.
	\end{flushleft}
\end{table}

\begin{table}[htbp]
	\centering
	\caption{Post-SVB event-window comparison across exposure proxies}
	\label{tab:LAProxySVBComparison}
	\resizebox{\textwidth}{!}{%
		\begin{tabular}{lcccccc}
			\hline
			& Retail & $LA^{reg}$ & $LA^{equal}$ & Retail & $LA^{reg}$ & $LA^{equal}$ \\
			& SVB & SVB & SVB & SVB + Unins. & SVB + Unins. & SVB + Unins. \\
			\hline
			$\sigma_{bt}$ & -0.019 & -0.002 & -0.002 & -0.010 & -0.003 & -0.004 \\
			& (0.021) & (0.003) & (0.003) & (0.017) & (0.003) & (0.003) \\
			$\sigma_{bt}\times\text{Retail Share}_{bt}$ & 0.018 &  &  & 0.008 &  &  \\
			& (0.021) &  &  & (0.017) &  &  \\
			$\sigma_{bt}\times LA^{reg}_{bt}$ &  & 0.002 &  &  & 0.001 &  \\
			&  & (0.003) &  &  & (0.003) &  \\
			$\sigma_{bt}\times LA^{equal}_{bt}$ &  &  & 0.000 &  &  & -0.000 \\
			&  &  & (0.004) &  &  & (0.004) \\
			$\sigma_{bt}\times PostSVB_t$ & -0.142$^{+}$ & -0.018 & -0.021 & -0.134 & -0.009 & -0.013 \\
			& (0.085) & (0.013) & (0.014) & (0.104) & (0.008) & (0.011) \\
			$\text{Retail Share}_{bt}\times PostSVB_t$ & -0.007 &  &  & -0.005 &  &  \\
			& (0.006) &  &  & (0.007) &  &  \\
			$LA^{reg}_{bt}\times PostSVB_t$ &  & -0.000 &  &  & -0.000 &  \\
			&  & (0.001) &  &  & (0.001) &  \\
			$LA^{equal}_{bt}\times PostSVB_t$ &  &  & -0.000 &  &  & -0.000 \\
			&  &  & (0.000) &  &  & (0.000) \\
			$\sigma_{bt}\times\text{Retail Share}_{bt}\times PostSVB_t$ & 0.141$^{+}$ &  &  & 0.138 &  &  \\
			& (0.084) &  &  & (0.100) &  &  \\
			$\sigma_{bt}\times LA^{reg}_{bt}\times PostSVB_t$ &  & 0.009 &  &  & 0.011 &  \\
			&  & (0.008) &  &  & (0.009) &  \\
			$\sigma_{bt}\times LA^{equal}_{bt}\times PostSVB_t$ &  &  & 0.002 &  &  & 0.005 \\
			&  &  & (0.010) &  &  & (0.010) \\
			$\sigma_{bt}\times PostSVB_t\times\text{Uninsured Share}_{bt}$ &  &  &  & 0.001 & -0.005 & -0.006 \\
			&  &  &  & (0.038) & (0.036) & (0.024) \\
			Observations & 8,909 & 8,712 & 8,712 & 7,766 & 7,643 & 7,643 \\
			$R^2$ & 0.044 & 0.049 & 0.049 & 0.057 & 0.061 & 0.061 \\
			Within $R^2$ & 0.001 & 0.000 & 0.000 & 0.001 & 0.001 & 0.000 \\
			Bank fixed effects & Yes & Yes & Yes & Yes & Yes & Yes \\
			Quarter fixed effects & Yes & Yes & Yes & Yes & Yes & Yes \\
			\hline
		\end{tabular}
	}
	\begin{flushleft}
		\footnotesize Notes: The dependent variable is the quarterly change in the cash-to-assets ratio. Standard errors are clustered two-way by bank and quarter. $PostSVB_t$ equals one in the post-Silicon Valley Bank event window. The level effect of $PostSVB_t$ is absorbed by quarter fixed effects. $^{+}p<0.10$.
	\end{flushleft}
\end{table}

\begin{table}[htbp]
	\centering
	\caption{Small-bank interaction regressions}
	\label{tab:SmallBankInteractions}
	\resizebox{\textwidth}{!}{%
		\begin{tabular}{lcccccc}
			\hline
			& Retail & Retail & $LA^{reg}$ & $LA^{reg}$ & $LA^{equal}$ & $LA^{equal}$ \\
			& Small & Small + Unins. & Small & Small + Unins. & Small & Small + Unins. \\
			\hline
			$\sigma_{bt}$ & -0.020 & -0.004 & -0.002 & -0.000 & -0.002 & -0.000 \\
			& (0.020) & (0.026) & (0.003) & (0.004) & (0.003) & (0.004) \\
			$\sigma_{bt}\times\text{Exposure}_{bt}$ & 0.019 & 0.006 & 0.001 & 0.000 & 0.000 & -0.001 \\
			& (0.019) & (0.025) & (0.003) & (0.003) & (0.004) & (0.004) \\
			$\sigma_{bt}\times SmallBank_{bt}$ & -0.062 & -0.071$^{***}$ & -0.003 & 0.003 & -0.002 & 0.002 \\
			& (0.090) & (0.016) & (0.007) & (0.023) & (0.009) & (0.022) \\
			$\sigma_{bt}\times\text{Exposure}_{bt}\times SmallBank_{bt}$ & 0.058 & 0.076$^{***}$ & 0.005 & 0.010$^{+}$ & 0.001 & 0.014$^{**}$ \\
			& (0.096) & (0.011) & (0.011) & (0.006) & (0.011) & (0.005) \\
			$\sigma_{bt}\times\text{Uninsured Share}_{bt}$ &  & -0.013$^{+}$ &  & -0.011 &  & -0.011 \\
			&  & (0.007) &  & (0.009) &  & (0.009) \\
			Observations & 8,909 & 7,766 & 8,712 & 7,643 & 8,712 & 7,643 \\
			$R^2$ & 0.043 & 0.058 & 0.049 & 0.061 & 0.049 & 0.062 \\
			Within $R^2$ & 0.000 & 0.001 & 0.000 & 0.001 & 0.000 & 0.001 \\
			Bank fixed effects & Yes & Yes & Yes & Yes & Yes & Yes \\
			Quarter fixed effects & Yes & Yes & Yes & Yes & Yes & Yes \\
			\hline
		\end{tabular}
	}
	\begin{flushleft}
		\footnotesize Notes: The dependent variable is the quarterly change in the cash-to-assets ratio. Standard errors are clustered two-way by bank and quarter. $SmallBank_{bt}=1$ for banks with less than \$1 billion in total assets and only domestic offices. $\text{Exposure}_{bt}$ is retail share, $LA^{reg}_{bt}$, or $LA^{equal}_{bt}$, depending on the column. $^{+}p<0.10$, $^{**}p<0.01$, $^{***}p<0.001$.
	\end{flushleft}
\end{table}

\begin{table}[htbp]
	\centering
	\caption{Recession-window robustness checks}
	\label{tab:PublicCallReportMacroRegression}
	\begin{tabular}{lccc}
		\hline
		& Macro dummies & Great Recession interactions & COVID interactions \\
		\hline
		Great Recession dummy & 0.000 &  &  \\
		& (0.001) &  &  \\
		COVID lockdown dummy & 0.000$^{+}$ &  &  \\
		& (0.000) &  &  \\
		$\sigma_{bt}\times LA_{bt}$ & 0.008 & 0.034 & 0.023 \\
		& (0.011) & (0.025) & (0.023) \\
		$\sigma_{bt}\times GreatRecession_t$ &  & 0.065 &  \\
		&  & (0.070) &  \\
		$LA_{bt}\times GreatRecession_t$ &  & 0.007 &  \\
		&  & (0.011) &  \\
		$\sigma_{bt}\times LA_{bt}\times GreatRecession_t$ &  & -0.067 &  \\
		&  & (0.077) &  \\
		$\sigma_{bt}\times COVID_t$ &  &  & 0.027 \\
		&  &  & (0.045) \\
		$LA_{bt}\times COVID_t$ &  &  & 0.010$^{+}$ \\
		&  &  & (0.006) \\
		$\sigma_{bt}\times LA_{bt}\times COVID_t$ &  &  & -0.032 \\
		&  &  & (0.043) \\
		Observations & 8,909 & 8,909 & 8,909 \\
		Bank fixed effects & Yes & Yes & Yes \\
		Quarter fixed effects & No & Yes & Yes \\
		\hline
	\end{tabular}
	\begin{flushleft}
		\footnotesize Notes: Standard errors are clustered two-way by bank and quarter. $^{+}p<0.10$. The pure Great Recession and COVID lockdown dummies are absorbed in the interaction specifications because those regressions include quarter fixed effects.
	\end{flushleft}
\end{table}

\clearpage
\section{Figures}
\begin{figure}[!htbp!]
	\centering
	\caption{Calibration of the loss-aversion index}
	\label{fig:LambdaCalibrationHalfCauchy}
	\includegraphics[width=0.78\textwidth]{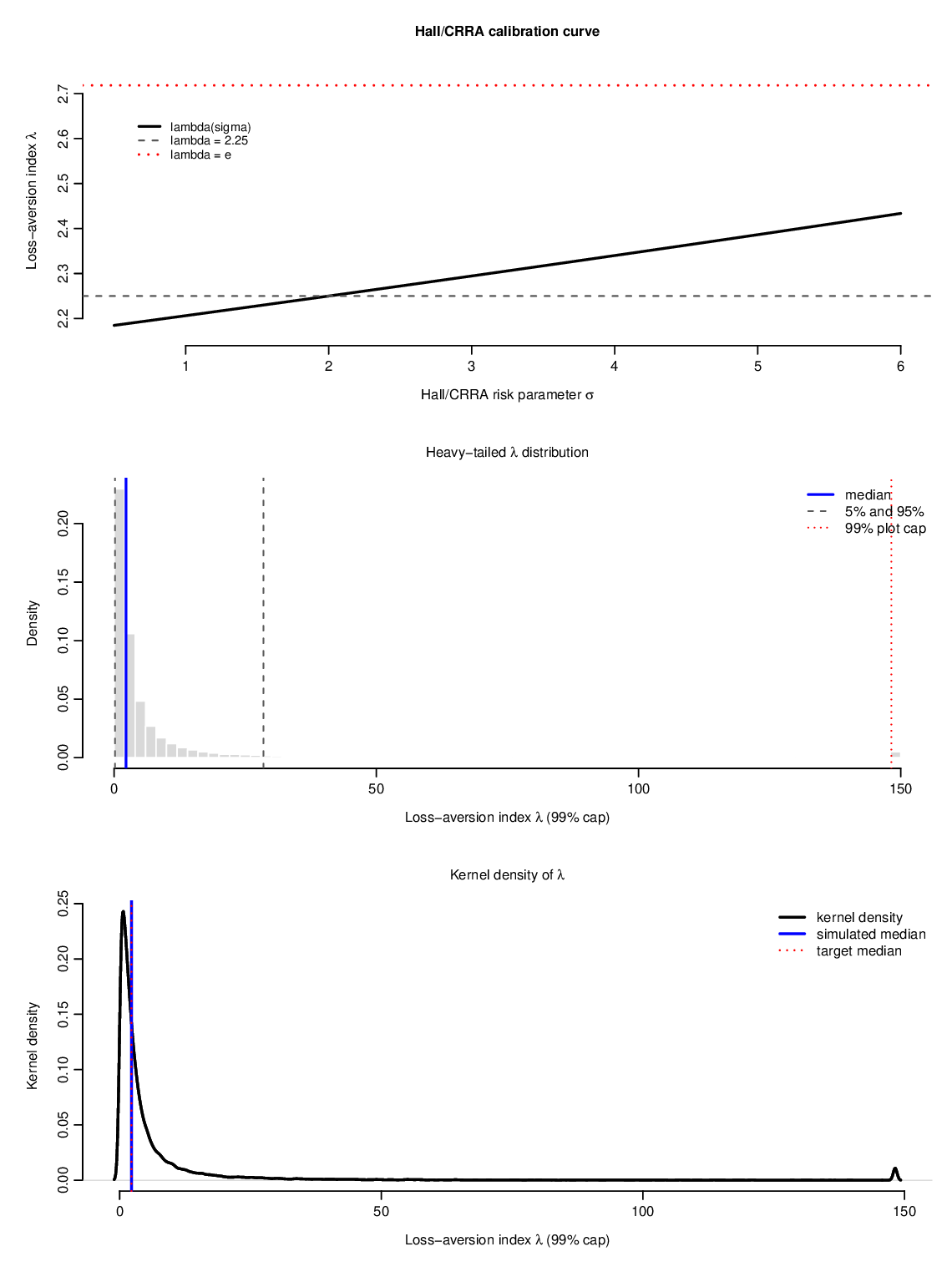}
	\par\medskip
	\begin{minipage}{0.88\linewidth}
		\footnotesize \emph{Note:} The upper panel plots the Hall/CRRA calibration curve for $\lambda$. The middle and lower panels show the simulated heavy-tailed distribution and kernel density obtained by multiplying the smooth Hall/CRRA component by a median-one half-Cauchy shock. The raw simulation preserves the heavy right tail; the displayed density is capped at the 99th percentile for readability.
	\end{minipage}
\end{figure}

\begin{figure}[htbp]
	\centering
	\vspace{-0.2em}
	\caption{Simulated bank-run exposure regressions}
	\label{fig:SimulatedBankRunExposureRegressions}
	\begin{subfigure}[t]{0.46\textwidth}
		\centering
		\includegraphics[width=\linewidth]{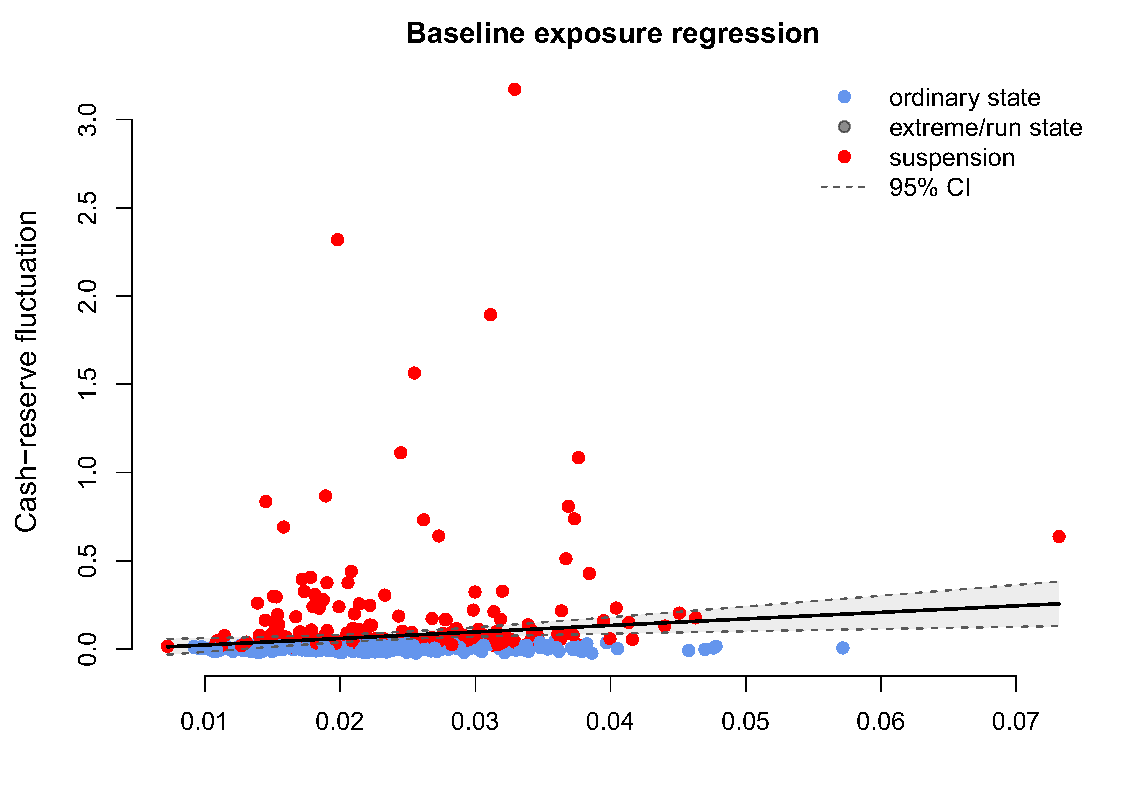}
		\caption{Baseline}
	\end{subfigure}
	\hfill
	\begin{subfigure}[t]{0.46\textwidth}
		\centering
		\includegraphics[width=\linewidth]{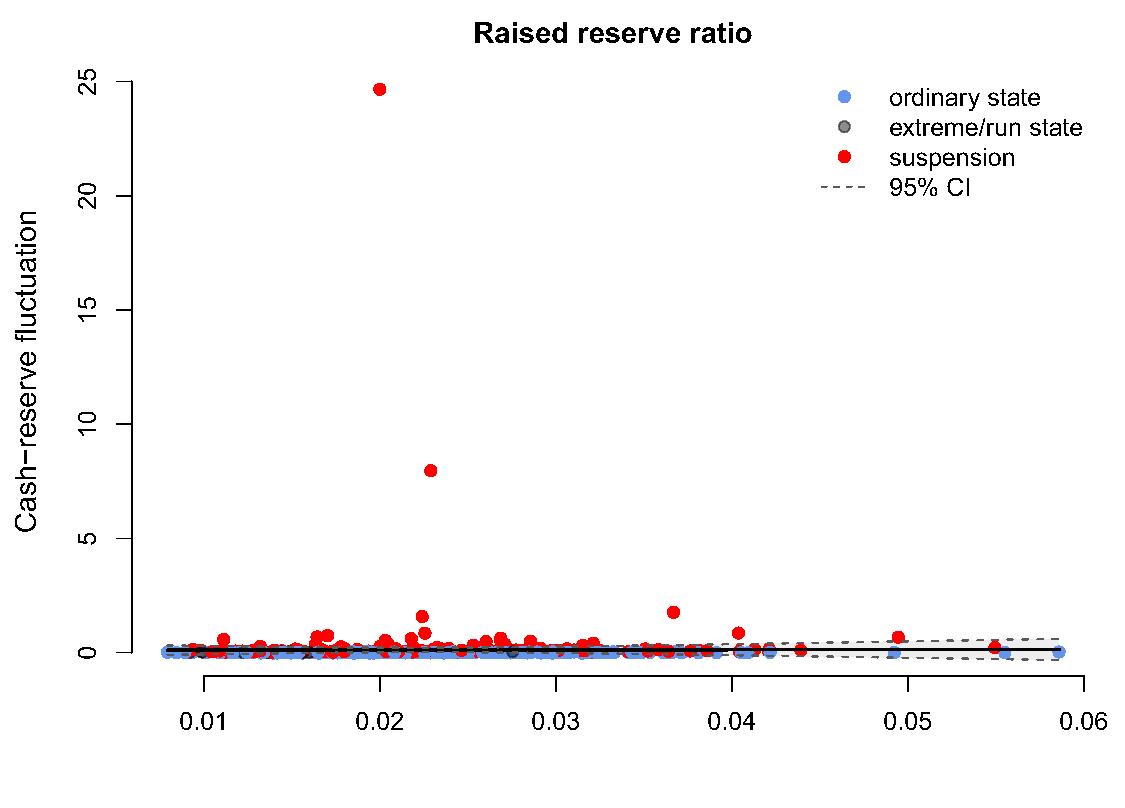}
		\caption{Raised reserve ratio}
	\end{subfigure}
	
	\vspace{0.12cm}
	\begin{subfigure}[t]{0.46\textwidth}
		\centering
		\includegraphics[width=\linewidth]{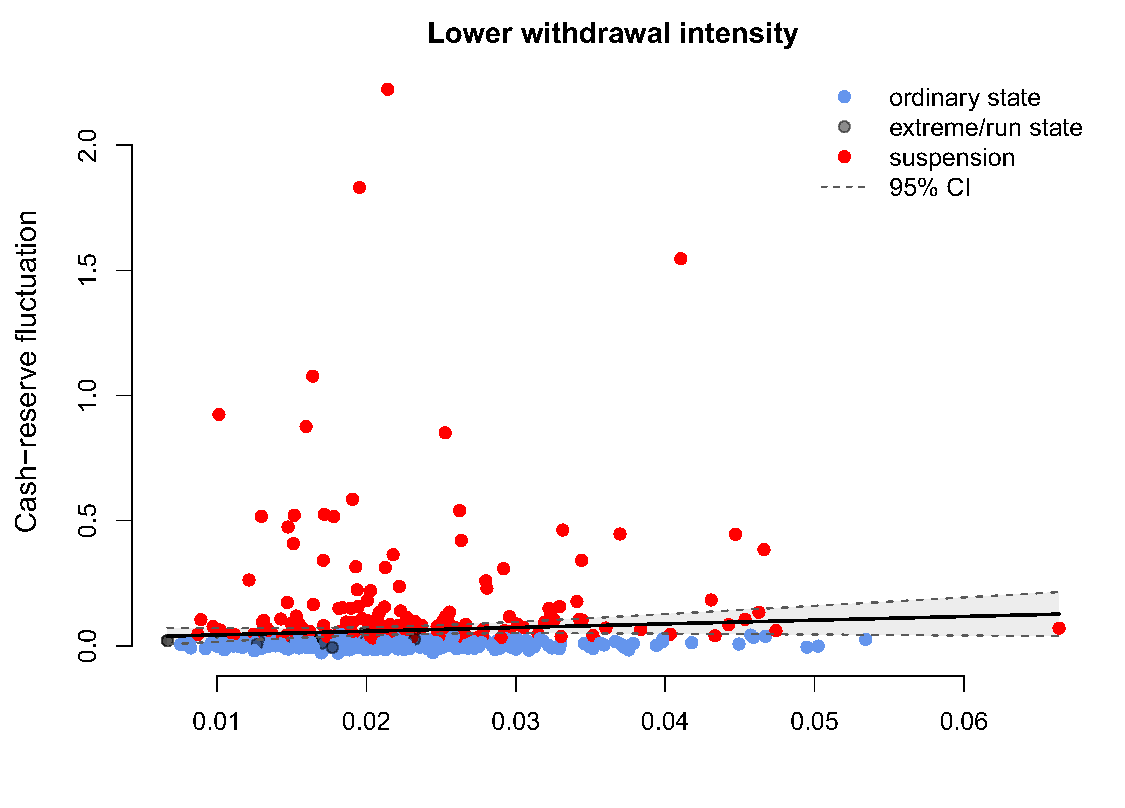}
		\caption{Lower withdrawal intensity}
	\end{subfigure}
	\hfill
	\begin{subfigure}[t]{0.46\textwidth}
		\centering
		\includegraphics[width=\linewidth]{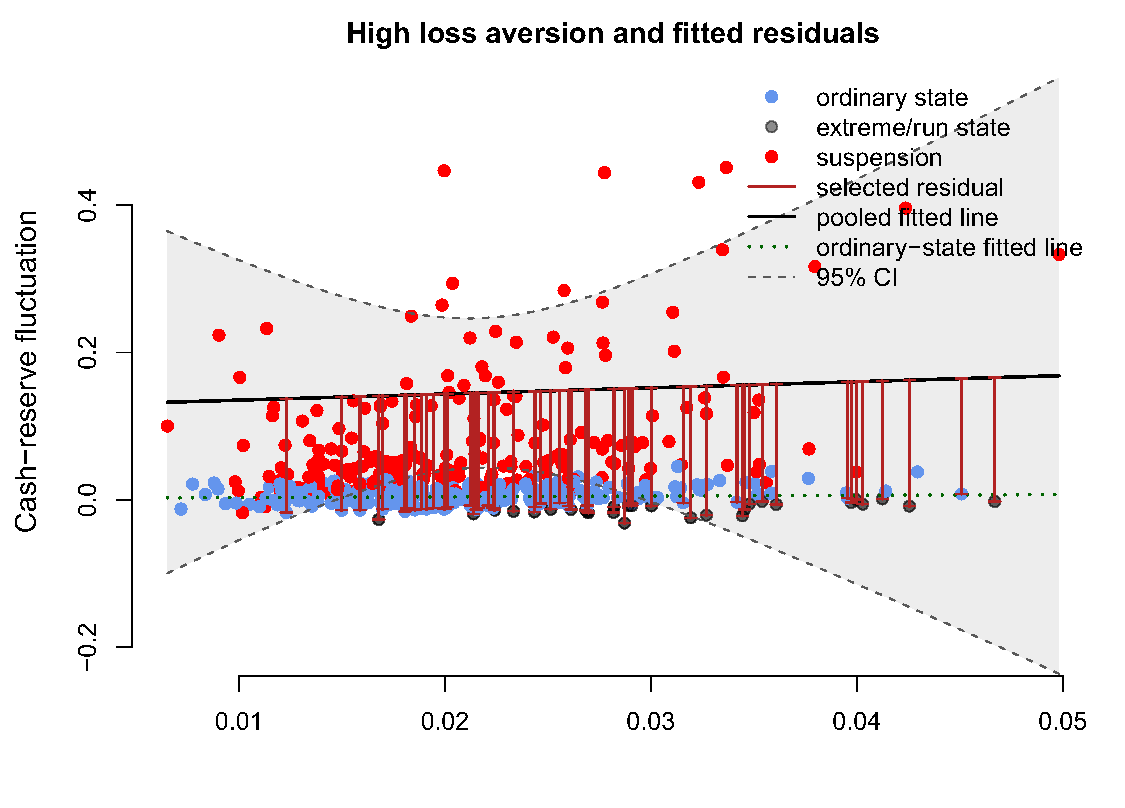}
		\caption{High loss aversion and fitted residuals}
	\end{subfigure}
	
	\vspace{0.25cm}
	\begin{minipage}{0.96\linewidth}
		\footnotesize \emph{Note:} Cornflower-blue observations are ordinary states, transparent black observations are extreme or run states, and red observations are suspension states. The dashed bands are the 95 percent confidence bands around the linear fit. The scenario panels show that the regression is useful as an operational diagnostic: the bank can ask whether reserve policy, withdrawal intensity, or unmodeled behavioral residuals are driving the stress-state observations. Simulation data, summary statistics, and diagnostic regressions are reported in the Internet Appendix.
	\end{minipage}
\end{figure}

\begin{figure}[htbp]
	\centering
	\caption{Reserve-need stability when latent loss aversion is omitted}
	\label{fig:OmittedLatentLossAversionReserveNeed}
	\includegraphics[width=0.90\textwidth]{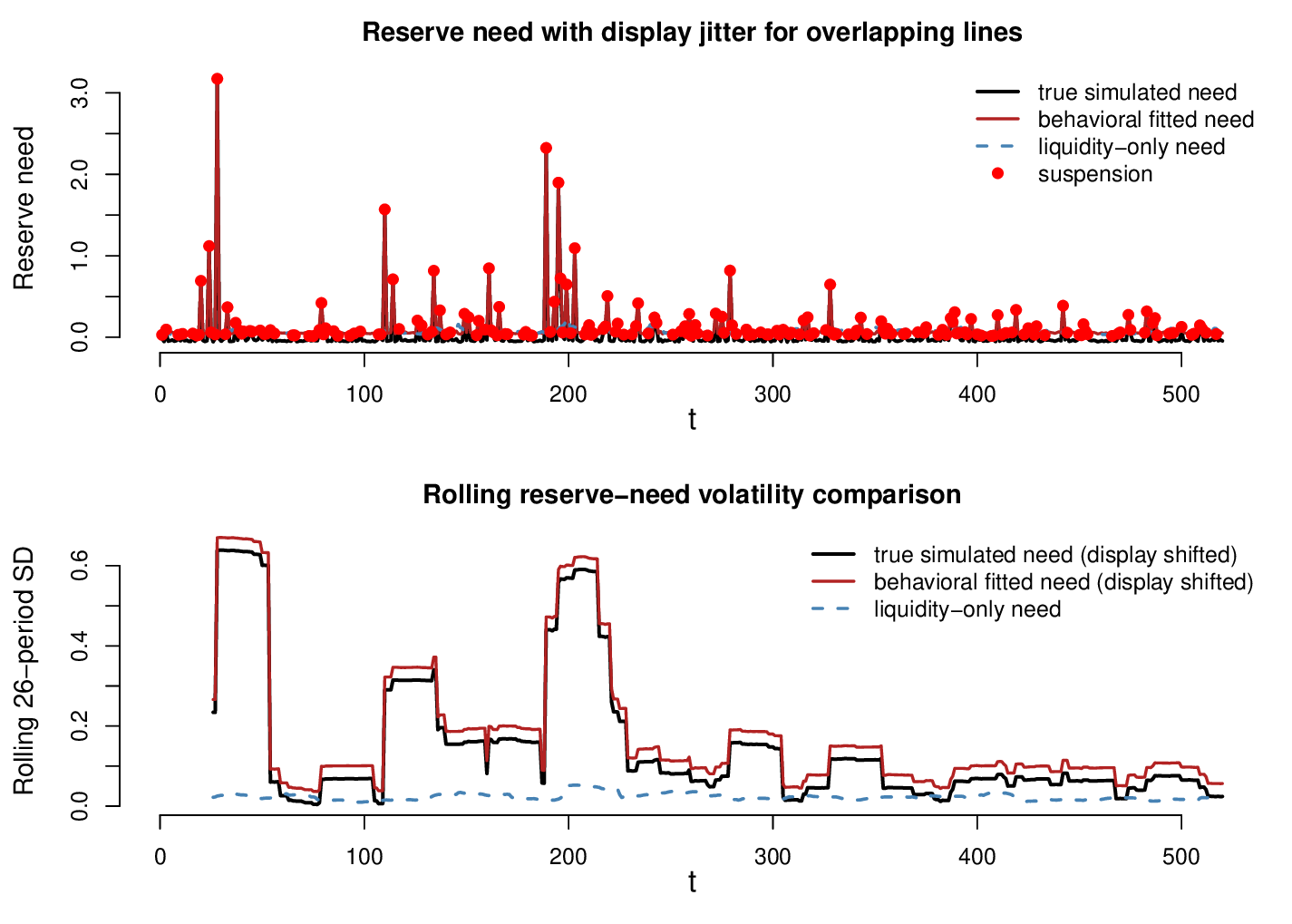}
	\par\medskip
	\begin{minipage}{0.9\linewidth}
		\footnotesize \emph{Note:} The upper panel overlays true simulated reserve need, behavioral fitted reserve need, and the liquidity-only estimate that omits latent loss aversion; the true and behavioral lines are shifted only slightly for display so that overlapping lines remain visible. The lower panel reports rolling reserve-need volatility. The liquidity-only specification makes reserve demand look too smooth and understates the volatility of reserve needs before suspension becomes visible.
	\end{minipage}
\end{figure}

\begin{figure}[htbp]
	\centering
	\caption{Simulated half-Cauchy loss-aversion SDE}
	\label{fig:LossAversionSDESimulation}
	\includegraphics[width=0.9\textwidth]{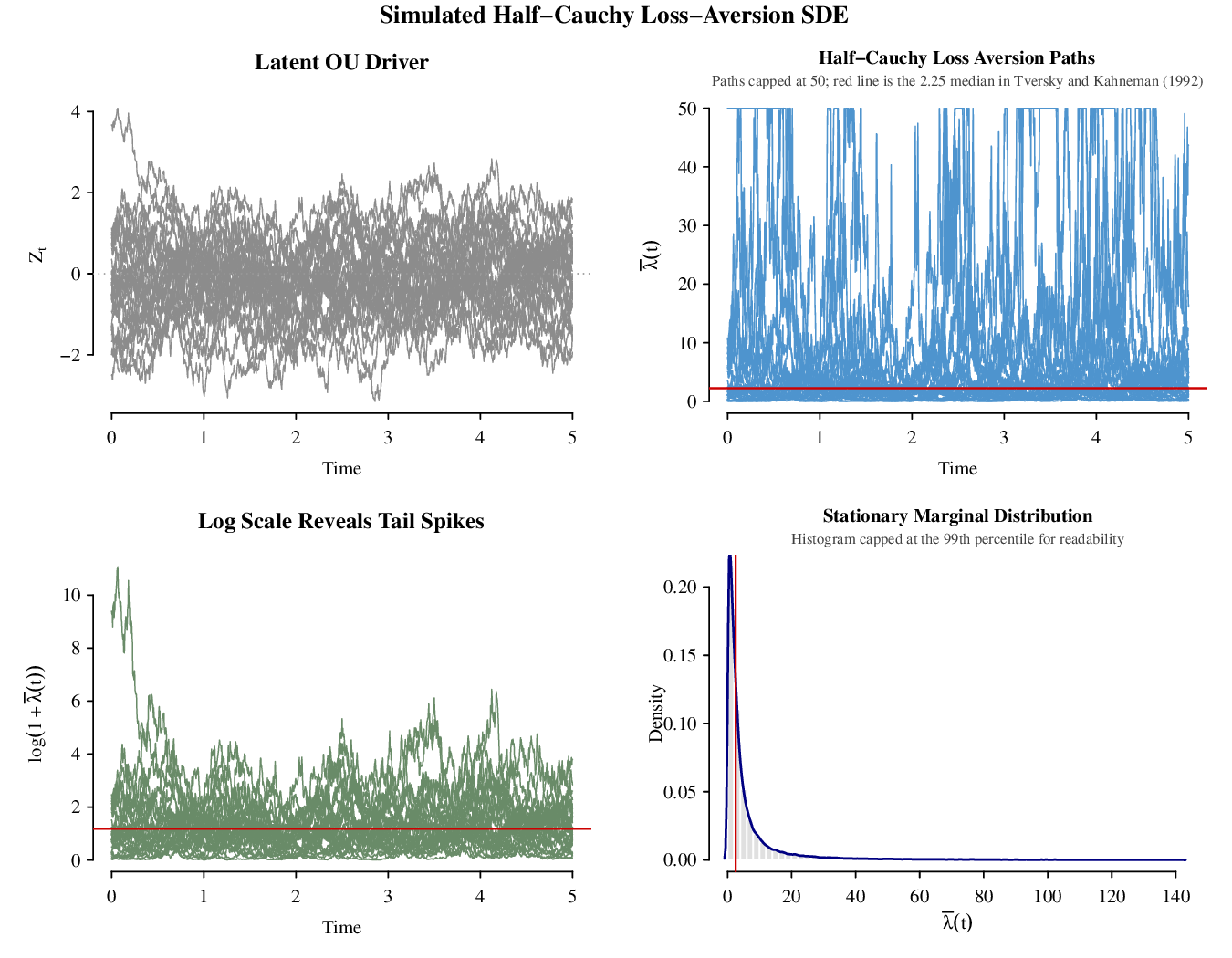}
	\par\medskip
	\begin{minipage}{1.0\linewidth}
		\footnotesize \emph{Note:} The upper-left panel plots the latent Ornstein-Uhlenbeck driver $Z_t$. The upper-right panel maps $Z_t$ into positive half-Cauchy loss-aversion paths $\bar{\lambda}(t)=m_\lambda\tan\{(\pi/2)\Phi(Z_t)\}$; paths are capped at $50$ for display. The red reference line is the median loss-aversion value $m_\lambda=2.25$ associated with \citep{TverKahn1992}. The lower panels report the log-scale paths and the stationary marginal distribution, making the tail spikes visible without losing the ordinary-state dynamics.
	\end{minipage}
\end{figure}

\begin{figure}[htbp]
	\centering
	\caption{Log bank-run exposure diagnostics}
	\label{fig:LogBankRunExposureSDE}
	\begin{subfigure}{0.48\textwidth}
		\centering
		\includegraphics[width=\textwidth]{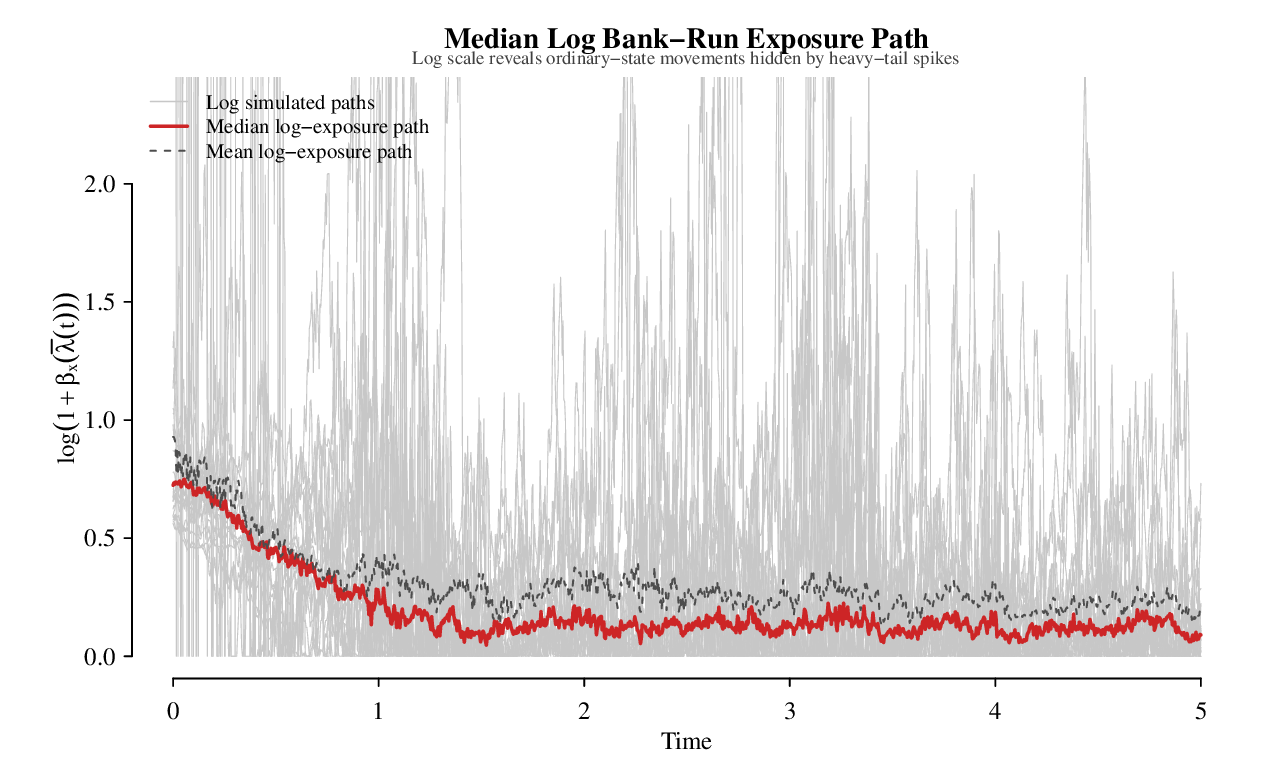}
		\caption{Median log-exposure path}
	\end{subfigure}\hfill
	\begin{subfigure}{0.48\textwidth}
		\centering
		\includegraphics[width=\textwidth]{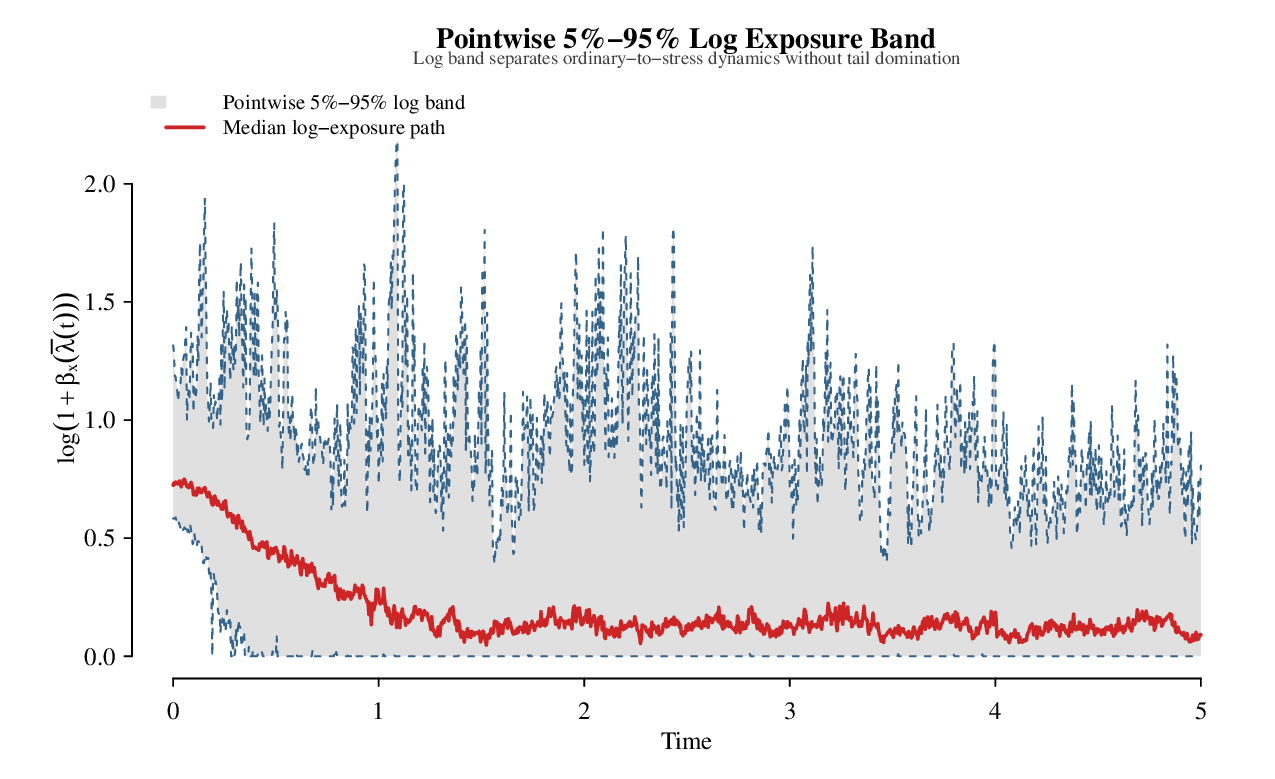}
		\caption{Pointwise $5$--$95$ percent log-exposure band}
	\end{subfigure}
	\par\medskip
	\begin{minipage}{1.0\linewidth}
		\footnotesize \emph{Note:} The left panel shows that the simulated median exposure path is relatively stable. The right panel shows that the pointwise interval can remain wide even when the median is stable, indicating that the exposure process can move between ordinary and stress states at any date.
	\end{minipage}
\end{figure}

\newpage
\singlespace
% \nocite{*}                       % dump the entire bibliography
\bibliographystyle{chicago}        % select a style
\section*{}
%\clearpage
\addcontentsline{toc}{section}{References} %To add to table of contents
\bibliography{BankRunTheory}         % bibliographic database

\end{document}